\documentclass[11pt,b5paper,twoside]{article}      

\newcommand{\mechyear}{2018}
\newcommand{\mechvol}{32}

\usepackage[cp1251]{inputenc}
\usepackage[T2A]{fontenc}
\usepackage[russian]{babel}

\usepackage[dvips]{graphicx}
\usepackage[dvips]{color}

\usepackage{amsmath}
\usepackage{amssymb}
\usepackage{amsxtra}
\usepackage{latexsym}
\usepackage{ifthen}
\usepackage[centerlast,footnotesize,sl]{caption2}

\usepackage{trudyn}

\usepackage{graphicx}
\def\KurganskyColSep{\setlength{\arraycolsep}{1.4pt}}

\newcommand{\Zed}{\mathbb Z}
\newcommand{\Nat}{\mathbb N}  
 
\newcommand{\Real}{\mathbb R}

\newcommand{\B}{\mathcal B}
\newcommand{\A}{\mathcal A}

\newcommand{\tB}{\acute{\mathcal B}}
\newcommand{\aB}{\acute{\mathcal B}}
\newcommand{\gB}{\grave{\mathcal B}}
\newcommand{\Univ}{\mathcal U_{{3}/{2}}}
\newcommand{\aUniv}{{\mathcal {\acute{U}}}_{{3}/{2}}}
\newcommand{\gUniv}{{\mathcal {\grave{U}}}_{{3}/{2}}}
\newcommand{\La}{\mathcal L}
\newcommand{\T}{\mathcal T}
\newcommand{\Sn}{\mathcal S}
\newcommand{\I}{\mathcal I}
\newcommand{\Bi}{\mathcal F} 

\newcounter{kurglemma}
\newcommand{\kurglemma}{\refstepcounter{kurglemma}}
\newcommand{\kurglemmaa}{\arabic{kurglemma}}

\newcounter{kurgtheorem}
\newcommand{\kurgtheorem}{\refstepcounter{kurgtheorem}}
\newcommand{\kurgtheorema}{\arabic{kurgtheorem}}

\newcounter{kurgcorollary}
\newcommand{\kurgcorollary}{\refstepcounter{kurgcorollary}}
\newcommand{\kurgcorollarya}{\arabic{kurgcorollary}}

\newcounter{kurgdefinition}

\newcommand{\qed}{}

\begin{document}


\udk{519.713, 511.4, 517.938} 

\title{ГРАФЫ ДЕ БРЁЙНА И СТЕПЕНИ ЧИСЛА $3/2$}
{Графы де Брёйна и степени числа $3/2$}



\titleeng{De Bruijn graphs and powers of $3/2$}

\author{А.Н.\,Курганский, И.\,Потапов}
\authoreng{O.\,Kurganskyy, I.\,Potapov}

\date{06.11.2018}




\address{ГУ <<ИПММ>>, Донецк\\ University of Liverpool, England}


\email{topologia@mail.ru}

\maketitle

\begin{abstract}
Рассматривается множество $\Zed^{\pm\omega}_{6}$ бесконечных в обе стороны слов $\xi$ в алфавите $\{0,1,2,3,4,5\}$ с разделяющим на целую левую $\lfloor\xi\rfloor$ дробную правую $\{\xi\}$ части знаком (запятой).
Для таких слов определено умножение на целые числа и деление на $6$ как умножение и деление столбиком в системе счисления $6$. 
В работе развивается конечно-автоматный подход анализа последовательностей вида 
$\left(\left\lfloor\xi\left(\frac{3}{2}\right)^n\right\rfloor \right)_{n\in\Zed}$ для слов $\xi\in\Zed^{\pm\omega}_{6}$ совпадающих по ряду свойств с $Z$-числами в $3/2$-проблеме Малера. Каждая такая последовательность $Z$-слов, записываемых друг под другом так, чтобы одинаковые разряды находились в одном столбце, представляет собой бесконечное $2$-мерное слово в алфавите $\Zed_6$.
Конечно-автоматное представление столбцов целой части таких $2$-мерных $Z$-слов обладают структурными свойствами графа де Брёйна. Такое представление доставляет ряд достаточных условий пустоты множества $Z$-чисел. 
Подход основывается на ряде результатов работы~\cite{KariKopra2018}, авторы которой применяют клеточные автоматы для анализа последовательностей $\left(\left\{\xi\left(\frac{3}{2}\right)^n\right\} \right)_{n\in\Zed}$, где $\xi\in\Real$.
\vspace{1mm}\\
\textbf{\emph {Ключевые слова: }}\itshape{распределение по модулю $1$, $Z$-числа, конечные автоматы, кусочно-аффинные отображения}
\end{abstract}

\abstracteng{
In this paper we consider the set $\Zed^{\pm\omega}_{6}$ of two-way infinite words $\xi$ over the alphabet $\{0,1,2,3,4,5\}$ with the integer left part $\lfloor\xi\rfloor$ and the fractional right part $\{\xi\}$ separated by a radix point. For such words, the operation of multiplication by integers and division by $6$ are defined as the column multiplication and division in base 6 numerical system.
The paper develops a finite automata approach for analysis of sequences
$\left (\left \lfloor \xi \left (\frac{3}{2} \right)^n \right \rfloor
\right)_{n \in \Zed}$ for the words $\xi \in \Zed^{\pm \omega}_{6}$ that have
some common properties with $Z$-numbers in Mahler's $3/2$-problem.
Such sequence of $Z$-words written under each other with the same digit positions in the same column is an infinite $2$-dimensional word over the alphabet $\Zed_6$. The automata representation of the columns in the integer part of $2$-dimensional $Z$-words has the nice structural properties of the de Bruijn graphs. This way provides some sufficient conditions for the emptiness of the set of $Z$-numbers.
Our approach has been initially inspirated by the proposition 2.5 in \cite{KariKopra2018} where authors applies cellular automata for analysis of $\left(\left\{\xi\left(\frac{3}{2}\right)^n\right\} \right)_{n\in\Zed}$, $\xi\in\Real$.
}


\kweng{modulo $1$ distribution, $Z$-numbers, finite automata, piecewise affine maps, De Bruijn graphs}
\KurganskyColSep
\section{Введение}
Известно~\cite{Weyl1916}, что последовательность $\{\xi\alpha^i\}$, $i\in\Nat$, равномерно распределена в интервале $[0,1)$ для почти всех положительных $\xi\in\Real$, где $\{x\}$ обозначает дробную часть числа $x$. При этом существуют числа $\xi$, для которых указанная последовательность ведёт себя иначе. Отсюда возникают задачи о свойствах $\{\xi\alpha^i\}$ для конкретных $\xi$ и $\alpha$, а также обратные задачи о существовании $\xi$ и $\alpha$, реализующих $\{\xi\alpha^i\}$ с заданными свойствами. Примером обратной задачи является открытая $3/2$-проблема Малера~\cite{Mahler1968} о существовании $Z$-чисел $\xi$ таких, что последовательность $\{\xi(3/2)^i\}$, $i\in\Nat$, полностью лежит в интервале $[0,1/2)$. 

В работе~\cite{KariKopra2018}, развивающей результаты работ~\cite{Akiyama2008, AkiyamaFrougnySakarovitch2008, Dubickas2008, FlattoLagariasPollington1995}, рассматривается множество $Z_{p/q}(S)$ таких $\xi$, что последовательность $\{\xi(3/2)^i\}$ полностью лежит в $S$, где $S\subseteq [0,1)$ является некоторым конечным объединением интервалов, $p,q\in\Nat$, и для него в контексте $3/2$-проблемы Малера ставится задача о поиске как можно большего множества $S$, при котором $Z_{p/q}(S)=\emptyset$, и как можно меньшего множества $S$, при котором $Z_{p/q}(S)\neq\emptyset$. Подходы к решению данной задачи основывается на исследовании клеточных автоматов, связь которых с $3/2$-проблемой Малера устанавливается в~\cite{Kari2012_1,Kari2012_2}.

В настоящей работе в рамках развития результатов работы~\cite{BournezKurgPotapov2018}  предлагается конечно-автоматный подход к проблематике $3/2$-проблемы Малера.
Рассматривается множество $\Zed^{\pm\omega}_{6}$ бесконечных в обе стороны слов $\xi$ в алфавите $\{0,1,2,3,4,5\}$ с разделяющим на условно целую $\lfloor\xi\rfloor$ (левую) и дробную $\{\xi\}$ (правую) части знаком (запятой).
Для таких слов определено умножение на целые числа и деление на $6$, содержательно совпадающее с умножением и делением столбиком в системе счисления $6$. 
В работе развивается конечно-автоматный метод анализа последовательностей вида 
$\left(\left\lfloor\xi\left(\frac{3}{2}\right)^n\right\rfloor \right)_{n\in\Zed}$ для слов $\xi\in\Zed^{\pm\omega}_{6}$ совпадающих по ряду свойств с $Z$-числами в $3/2$-проблеме Малера. Такая последовательность $Z$-слов, записанных друг под другом так, чтобы одинаковые разряды находились в одном столбце, представляет собой бесконечное $2$-мерное слово в алфавите $\Zed_6$.
Конечно-автоматное представление столбцов целой части таких $2$-мерных $Z$-слов обладают структурными свойствами графа де Брёйна. Такое представление доставляет ряд наглядных достаточных условий пустоты множества $Z$-чисел. 
Подход основывается на идее предложения 2.5 работы~\cite{KariKopra2018}, авторы которой применяют клеточные автоматы для анализа последовательностей $\left(\left\{\xi\left(\frac{3}{2}\right)^n\right\} \right)_{n\in\Zed}$, где $\xi\in\Real$. Обобщенная форма предложения 2.5 работы~\cite{KariKopra2018} в настоящей работе представлена теоремой~\ref{theorem:FigureKurgPotap2018_triangle}.

\section{Обозначения и определения}
Как обычно, через $\mathbb {N}$, $ \mathbb {Z}$,  $ \mathbb {P} $, $ \mathbb {Q} $ and $ \mathbb {R }$ обозначаем натуральные, целые, простые, рациональные и действительные числа. Используем следующие сокращения и обозначения: $\Zed_n=\Zed^+_n=\{0,1,2,\ldots,n-1\}$, $\Zed^-_n=\{-(n-1),(n-2),\ldots,-2,-1\}$, $\mathbb{Z}^+ = \{0,1,2, \ldots \}$, $\Zed^- = \{\ldots,-3,-2,-1\} $, $\Zed^\pm = \Zed$.

Фигурные скобки используются для записи множеств, а также в записи $\{x\}$ дробной части числа $x$. Через $\left\lfloor x \right\rfloor$ обозначаем наибольшее целое число меньшее $x$. 
Через $2^M$ обозначаем булеан множества $M$. Если $f:X\rightarrow Y$, $g:Y\rightarrow Z$, $h:X'\rightarrow Y'$, то $gf$ обозначает суперпозицию функций: $gf:X\rightarrow Z$,  а $f\times h$ декартово произведение функций. Если $A\subseteq X$, то $f(A)=\{f(x)|x\in A\}$.

Для работы с $1$-мерными и $2$-мерными словами будут использоваться следующие средства выделения их фрагментов с помощью индексов, частично заимствованные из языка Python, инструментами которого проводились вычислительные эксперименты данной работы. Пусть $a$, $b$, $c$, $d\in\Zed$,  $a<b$,  $c<d$. Определим следующие функции: 
\[
\begin{aligned}
&[a:b]:\Zed_{b-a}\rightarrow \{a,a+1,\ldots,b-1\}, & [a:b](x)=x+a,\\
&[a:+\infty]:\Zed^+\rightarrow\{x\in\Zed|x\ge a\}, & [a:+\infty](x)=x+a, \\
&[-\infty:b]:\Zed^-\rightarrow\{x\in\Zed|x<b\},    & [-\infty:b](x)=x+b, \\
&[a:b,c:d]=[a:b]\times [c:d].
\end{aligned}
\]
Через $[:]$ обозначим тождественное отображение на $\Zed^{\pm}$.
Также используем сокращения записи: $[a]=[a:a+1]$, $[:b]=[-\infty:b]$, $[a:]=[a:+\infty]$.

Множество $\Zed^2=\Zed\times\Zed$ рассматриваем как модуль над кольцом $\Zed$.
Элементы $\Zed^2$ представляем вектор-столбцами.
В отображениях вида $f:\Zed\rightarrow X$, $f:\Zed^2\rightarrow X$ элементы из $\Zed$ и $\Zed^2$ называем координатами их образов в $X$.

\paragraph{Одномерные слова.}
Под одномерным словом в алфавите $A$ мы понимаем то же, что и в теории формальных языков. Однако, поскольку в работе речь идёт о словах в алфавите цифр, то есть о числах, в том числе с разделяющим целую и дробную части знаком, то слово $w\in A^*$ длины $n$ мы также будем интерпретировать как функцию $w:\Zed_n\rightarrow A$. Например, пусть $A=\Zed_{10}$ и $w=314159$. Тогда $w(0)=9$, $w(5)=3$, а $w[1:4]=415$. Под бесконечными словами мы понимаем отображения: $w:\Zed\rightarrow A$ (бесконечное в обе стороны), $w:\Zed^+\rightarrow A$ (бесконечное влево), $w:\Zed^-\rightarrow A$ (бесконечное вправо).
К примеру, число $w=3.14159$ мы понимаем как отображение из $\Zed_1\cup\Zed^-_6$ в алфавит цифр. В частности, $w(-2)=4$. Слова любой размерности в алфавите $A$ будем называть $A$-словами. Поскольку в работе речь идёт в основном о $\Zed_6$-словах, то мы будем опускать указание на алфавит, когда он понятен из контекста. 
\paragraph{Двумерные слова.}
Пусть $K\subseteq \Zed^2$. Под $2$-мерным словом $u$ в алфавите $A$ понимаем любое отображение
$u:K\rightarrow A$. 
Если $K$ конечное или бесконечное, то $u$ называется, соответственно, конечным или бесконечным.

Под конечным $2$-мерным прямоугольным словом $u$ размера $n\times m$ в алфавите $A$ понимаем отображение $u:\Zed_n\times\Zed_m\rightarrow A$. 
Так как в работе под словами понимаются числовые объекты, то мы используем такую индексацию цифр слова $u$:
\[
u=
\begin{pmatrix}
u_{0,m-1}&u_{0,m-2}&\ldots&u_{0,0}\\
u_{1,m-1}&u_{1,m-2}&\ldots&u_{1,0}\\
\hdotsfor{4}\\
u_{n-1,m-1}&u_{n-1,m-2}&\ldots&u_{n-1,0}
\end{pmatrix}.
\]
Понимаем, что $u(i,j)=u[i,j]=u_{i,j}$, при этом $u(i,j)$ является цифрой, а $u[i,j]$ однобуквенным словом. 

Бесконечными прямоугольными $2$-мерными словами являются следующие отображения: 
$
u:\Zed_n\times\Zed^\sigma\rightarrow A
$, 
$
u:\Zed^\sigma\times\Zed_m\rightarrow A
$,
$
u:\Zed^{\sigma}\times\Zed^{\sigma'}\rightarrow A
$,
где $\sigma, \sigma'\in\{-,+,\pm\}$. 
Прямоугольным подсловом (фрагментом) размера $n\times m$ $2$-мерного слова $u$ называем суперпозицию $u[i:i+n,j:j+m]$ отображения $u$ и отображения $[i:i+n,j:j+m]$ для некоторых $i,j\in\Zed$.

Пусть $x\in\Zed_6$. Под $x^{n}$, $x^{\pm\omega}$, $x^{+\omega}$, $x^{-\omega}$ понимаем следующие всюду определенные слова:
$x^n:\Zed_n\rightarrow \{x\}$,
$x^{\pm\omega}:\Zed^{\pm}\rightarrow \{x\}$,
$x^{-\omega}:\Zed^{-}\rightarrow \{x\}$,
$x^{+\omega}:\Zed^{+}\rightarrow \{x\}$.
Например, $0^5=00000$, $0^{+\omega}=\ldots 000$, $0^{-\omega}=000\ldots$, $0^{\pm\omega}=\ldots 0,0\ldots$, где запятая в выражении для $0^{\pm\omega}$ разделяет условные дробную (справа) и целую (слева) части.
\begin{remark}
Запись слова $w=a_{n-1}a_{n-1}\ldots a_0$, $a_i\in A$, мы понимаем как функцию $w:\Zed_n\rightarrow A$ такую, что $w(i)=w[i]=a_i$. В строковой записи $w=a_{n-1}a_{n-1}\ldots a_0$ слова $w:\Zed_n\rightarrow A$ начало координат мы приписываем её крайнему правому символу. Таким образом, мы используем принятую для целых чисел нумерацию цифр. Исключения будут оговариваться. Например, одним из способов указания на другую нумерацию будет разделяющий на целую и дробную часть числа знак. Исходя из сказанного, должно быть понятно, если это важно, как происходит нумерация цифр при конкатенации или сцеплении слов. Операция сцепления слов определена ниже.
\end{remark}

Под транспонированием слова $u$ понимаем такое слово $u^T$, что $u^T[i,j]=u[j,i]$.
Вообще, для прямоугольных $2$-мерных слов заимствуем термины из языка матриц.
\paragraph{Горизонтальное сцепление слов.}
Наряду с конкатенацией в работе используется операция сцепления слов. 
Пусть слова $u = (v_1v_2\ldots v_n)$ и $u' = (w_1w_2\ldots w_m)$ имеют соответственно размеры $r\times n$ и $r\times m$, где $v_i$ и $w_j$ $r$-мерные столбцы. Если $v_n = w_1$, то 
$u\circ u' =  (v_1v_2\ldots v_n w_2\ldots w_m)$,
иначе $u\circ u'$ не определено. Операцию $\circ$ назовём горизонтальным сцеплением слов.
Для примера:
$
\left(\begin{smallmatrix}
1&2\\
4&5
\end{smallmatrix}\right)
\circ
\left(\begin{smallmatrix}
2&3\\
5&6
\end{smallmatrix}\right)
=
\left(\begin{smallmatrix}
1&2&3\\
4&5&6
\end{smallmatrix}\right)
$.
Пусть $A$ и $B$ множества слов, тогда
$
A\circ B = \{a\circ b|a\in A, b\in B\}
$ и 
\[
A^{1\circ} = A,\quad
A^{n\circ} = A^{(n-1)\circ}\circ A,
\]
\[
A^{-\omega\circ} = A\circ A\circ\ldots,
\quad
A^{+\omega\circ} = \ldots\circ A\circ A,
\]
\[
A^{\pm\omega\circ} = A^{+\omega\circ}\circ A^{-\omega\circ} =\ldots\circ A\circ A\circ\ldots.
\]
В $A^{\pm\omega\circ}$ по умолчанию, если не оговаривается другое, подразумеваем привязку к нулевому разряду крайнего правого символа $A^{+\omega\circ}$.
\paragraph{Вертикальное сцепление слов.}
Пусть слова $u$ и $u'$ имеют размеры $n\times r$ и $m\times r$. При этом последняя строка слова $u$ совпадает с верхней строкой слова $u'$. Тогда вертикальное сцепление $u\bullet u'$ определяется так: $u\bullet u' =  (u'^T\circ u^T)^T$.
Для примера:
$
\left(\begin{smallmatrix}
3&4\\
5&6
\end{smallmatrix}\right)
\bullet
\left(\begin{smallmatrix}
1&2\\
3&4
\end{smallmatrix}\right)
=
\left(\begin{smallmatrix}
1&2\\
3&4\\
5&6
\end{smallmatrix}\right)
$.
Как и для $\circ$ определяем обозначения $A\bullet B$, $A^{n\bullet}$, $A^{-\omega\bullet}$, $A^{+\omega\bullet}$, $A^{\pm\omega\bullet}$.
\paragraph{Отношения одномерных слов в двумерных словах.}
Всякое множество двумерных слов $W$ устанавливает отношения между столбцами и строками. Например, множество $W$ двумерных слов вида $w:\Zed_2\times\Zed_n\rightarrow A$ устанавливает бинарное отношение $R$ на строках в алфавите $A$: $(w,w')\in R$, если $\left(\begin{smallmatrix}w\\ w'\end{smallmatrix}\right)\in W$. Аналогично, множество $W$ двумерных слов вида $w:\Zed_n\times\Zed_2\rightarrow A$ устанавливает бинарное отношение на столбцах в алфавите $A$. В работе не будет использоваться дополнительное обозначение для таких отношений. Само множество $W$ будет интерпретироваться в зависимости от контекста как то или иное, но в явном виде оговариваемое, отношение между одномерными словами.

\paragraph{Универсальное множество слов $\Univ$.}
Для бесконечных в обе стороны $\Zed_6$-слов формально определяем операцию сложения и вычитания столбиком, а также умножение столбиком на целые числа и деление на $6$ как это делается для чисел в $6$-ричной системе счисления.
Через $\Univ$ обозначаем множество всех 2-мерных слов $u:\Zed\times\Zed\rightarrow \Zed_6$, строки которых представляют собой последовательности $1$-мерных $\Zed_6$-слов, получаемых друг из друга умножением предыдущего на $3/2$:
{\small
\[
\Univ = \{u:\Zed^2\rightarrow\Zed_6|\forall i\in\Zed\quad u[i+1,:]=\frac{3}{2}u[i,:]\},
\]
}
Обозначим: 
{\small
\[
\Univ^{n\times m} = \left\{u[i:i+n,j:j+m]\left|u\in \Univ, i,j\in\Zed\right.\right\},\quad
\B= \Univ^{2\times 2},\quad
\I= \Univ^{2\times 1}.
\]
}
\begin{definition}
Элементы множеств $\Univ^{2\times n}$ и $\Univ^{n\times 2}$ назовём, соответственно,  горизонтальными и вертикальными $n$-мерными продукционными парами. 
\end{definition}

Пусть $A$~-- преобразование параллельного переноса пространства $\Zed^2$, то есть для некоторого $a\in \Zed^2$ и любого $x\in \Zed^2$ $A(x)=x+a$.
Тогда для любого $u\in \Univ$ верно $uA\in \Univ$. Отсюда справедливы леммы~\ref{lemma:svoistvoU1},~\ref{lemma:svoistvoU2}.
\kurglemma\label{lemma:svoistvoU1}
\begin{lemma}[\arabic{kurglemma}]
Выполняется равенство: $\Univ^{n\times m} = \left\{u[0:n,0:m]\left|u\in \Univ\right.\right\}$.
\end{lemma}
\kurglemma\label{lemma:svoistvoU2}
\begin{lemma}[\arabic{kurglemma}]
Для всех $i,j\in\Zed$ $\Univ[:,i]=\Univ[:,j]$, $\Univ[i,:]=\Univ[j,:]$.
\end{lemma}

Пусть линейные преобразования $A:\Zed^2\rightarrow \Zed^2$ и $B:\Zed^2\rightarrow \Zed^2$ заданы матрицами:
$
A=
\left(\begin{smallmatrix}
1&0\\
1&1
\end{smallmatrix}\right)$, $B=
\left(\begin{smallmatrix}
1&0\\
-1&1
\end{smallmatrix}\right)$,
то есть
$
A
\left(\begin{smallmatrix}
x\\
y
\end{smallmatrix}\right)=
\left(\begin{smallmatrix}
1&0\\
1&1
\end{smallmatrix}\right)
\left(\begin{smallmatrix}
x\\
y
\end{smallmatrix}\right)=
\left(\begin{smallmatrix}
x\\
x+y
\end{smallmatrix}\right)$,
$B
\left(\begin{smallmatrix}
x\\
y
\end{smallmatrix}\right)=
\left(\begin{smallmatrix}
x\\
y-x
\end{smallmatrix}\right)$. 
Зафиксируем обозначения:
{\small
\[
\aUniv = \left\{uA|u\in \Univ\right\},\quad
\gUniv^{n\times m} = \left\{u[0:n,0:m]\left|u\in \gUniv\right.\right\},\quad
\aB= \aUniv^{2\times 2},
\]
\[
\gUniv = \left\{uB|u\in \Univ\right\},\quad
\aUniv^{n\times m} = \left\{u[0:n,0:m]\left|u\in \aUniv\right.\right\},\quad
\gB= \gUniv^{2\times 2},
\]
\[
\B^{0\circ}=\aB^{0\circ}=\gB^{0\circ}=\I.
\]}
\begin{definition}
Если для слов $u,w:\Zed^2\rightarrow\Zed_6$ и некоторого аффинного преобразования $A:\Zed^2\rightarrow\Zed_6$ выполняется $w=uA$, то называем слова $u$ и $w$ конгруэнтными и пишем $u\sim w$. Отношение $\sim$ является эквивалентностью. Пусть $U$, $W$ множества слов. Если для любого $u\in U$ существует $w\in W$ такое, что $u\sim w$, и наоборот, то пишем $U\sim W$.
\end{definition}
\begin{definition}
Элементы множеств $\aUniv^{2\times n}$ и $\gUniv^{2\times n}$ назовём диагонально, соответственно, восходящими и нисходящими $n$-мерными продукционными парами. Иногда элементы из $\aUniv^{2\times n}$ просто называем диагональными продукционными парами.
\end{definition}

\paragraph{$V$-таблицы и $H$-таблицы.}
Используем следующее обозначение:
{\small
\begin{equation}\label{formula:2}
\begin{bmatrix}
a_1 & b_1 & \ldots & c_1 \\
a_2 & b_2 & \ldots & c_2 \\
\hdotsfor{4} \\
a_m & b_m & \ldots & c_m \\
\hline
a_{m+1} & b_{m+1} & \ldots & c_{m+1} \\
a_{m+2} & b_{m+2} & \ldots & c_{m+2} \\
\hdotsfor{4} \\
a_{n} & b_{n} & \ldots & c_n \\
\end{bmatrix}
 = 
\left\{\left.
\begin{pmatrix}
a_i & b_i & \ldots & c_i\\
a_j & b_j & \ldots & c_j
\end{pmatrix}
\right|\begin{matrix} 1\le i\le m,\\ m<j\le n \end{matrix} 
\right\}
\end{equation}
\begin{equation}\label{formula:1}
\left[\begin{matrix}
a_1 & a_2 & \ldots & a_m & \vrule \\
b_1 & b_2 & \ldots & b_m & \vrule \\
\hdotsfor{4} & \vrule  \\
c_1 & c_2 & \ldots & c_m & \vrule 
\end{matrix}
\begin{matrix}
&a_{m+1} & a_{m+2} & \ldots & a_n \\
&b_{m+1} & b_{m+2} & \ldots & b_n \\
&\hdotsfor{4} \\
&c_{m+1} & c_{m+2} & \ldots & c_n \\
\end{matrix}\right]
 = 
\left\{\left.
\begin{pmatrix}
a_i & a_j \\
b_i & b_j \\
\hdotsfor{2} \\
c_i & c_j
\end{pmatrix}
\right|\begin{matrix} 1\le i\le m,\\ m<j\le n \end{matrix} 
\right\}
\end{equation}
}

Представление декартова произведения столбцов в форме левой части в~(\ref{formula:1}) назовём $V$-таблицей, а декартово произведение строк в форме левой части в~(\ref{formula:2}) назовём $H$-таблицей. В $H$-таблице $T$ множество строк над разделяющей чертой (соответственно: под чертой) назовём верхней (нижней) компонентой таблицы и будем обозначать $\pi_U(T)$ ($\pi_O(T)$).  В $V$-таблице $T$ множество столбцов слева от разделяющей черты (соответственно: справа) назовём левой (правой) компонентой таблицы и будем обозначать $\pi_L(T)$ ($\pi_R(T)$).  

Всякое конечное множество $W$ горизонтальных продукционных пар представимо в виде объединения некоторых $H$-таблиц. Объединение таблиц мы обозначаем простым приписыванием их друг к другу, как показано, например, в леммах~\ref{lemma:Itable}, \ref{lemma:aB}, \ref{lemma:B} и т.д. Если каждая строка в верхних (нижних) компонентах таблиц встречается только в одной $H$-таблице, то такое представление множества $W$ назовём представлением в виде $H$-таблиц с уникальными верхними (нижними) компонентами. Аналогично, для множеств вертикальных продукционных пар используем представление $V$-таблицами с уникальными левыми и правыми компонентами.

Представляя множество продукционных пар с помощью $H$-таблиц с уникальной верхней (нижней) компонентой, мы подразумеваем, что различные таблицы различаются в нижних (верхних) компонентах, то есть нет таких двух таблиц, которые можно представить одной. Аналогичное подразумевается для $V$-таблиц.

Поскольку $V$-таблицы и $H$-таблицы являются формой представления множеств, соответственно, вертикальных и горизонтальных продукционных пар, то для $H$-таблиц очевидным образом определена операции горизонтального сцепления $\circ$, а для $V$-таблиц операция вертикального сцепления $\bullet$.

\begin{definition}
Пусть $W$ множество двумерных слов. Слово $w\in W$ назовём горизонтальным тупиком в $W$, если $W^{+\omega\circ}\circ w\circ W^{-\omega\circ} = \emptyset$. Слово $w\in W$ назовём вертикальным тупиком в $W$, если 
$W^{+\omega\bullet}\bullet w\bullet W^{-\omega\bullet} = \emptyset$.
\end{definition}

\kurglemma
\begin{lemma}[\kurglemmaa]\label{lemma:tupikiU2n}
$\Univ^{n\times m}$ не содержит горизонтальных и вертикальных тупиков. 
\end{lemma}
\begin{proof}
Прямо следует из определения $\Univ^{n\times m}$.\qed
\end{proof}

\kurglemma
\begin{lemma}[\kurglemmaa]\label{lemma:tupikiWbulletW}
Если $W$ не содержит горизонтальных (вертикальных) тупиков, то $W^{n\circ}$ ($W^{n\bullet}$) не содержит горизонтальных  (вертикальных) тупиков, $n\in\Nat$.
\end{lemma}
\begin{proof}
Непосредственно следует из определений.\qed
\end{proof}

\section{$n$-Мерные продукционные пары}
Пусть двумерное $\Zed_6$-слово
\[
\begin{matrix}
a_6&a_5&a_4&a_3&a_2&a_1&a_0 \\
	 &b_5&b_4&b_3&b_2&b_1&	 \\
	 &   &c_4&c_3&c_2&   &	 \\
	 &   &   &d_3&   &   &	 
\end{matrix}
\]
такое, что каждая её последующая строка получена из предыдущей умножением на $3/2$ и удалением крайних цифр согласно схеме слова. Леммы~\ref{lemma:Itable}, \ref{lemma:aB}, \ref{lemma:B}, \ref{lemma:gB} и \ref{lemma:U32} получены перебором всех слов $a_6a_5a_4a_3a_2a_1a_0$  в алфавите $\Zed_6$.

\kurglemma
\begin{lemma}[\kurglemmaa]\label{lemma:Itable}
Множество $\I$ в форме $H$-таблиц с уникальными нижними компонентами имеет вид:
{\scriptsize
\[
\begin{aligned}
\begin{bmatrix}
0\\
2\\
4\\
\hline
0\\
3
\end{bmatrix}
\begin{bmatrix}
1\\
3\\
5\\
\hline
2\\
5
\end{bmatrix}
\begin{bmatrix}
0\\
2\\
4\\
1\\
3\\
5\\
\hline
1\\
4
\end{bmatrix}
\end{aligned}
\]
}
\end{lemma}
\kurglemma
\begin{lemma}[\kurglemmaa]\label{lemma:aB}
Множество $\aB$ в форме $H$-таблиц с уникальными нижними компонентами имеет вид:
{\scriptsize
\[
\begin{aligned}
\begin{bmatrix}
0&0\\
2&2\\
4&4\\
\hline
0&0\\
3&0
\end{bmatrix}
\begin{bmatrix}
0&1\\
2&3\\
4&5\\
\hline
0&2\\
3&2
\end{bmatrix}
\begin{bmatrix}
0&2\\
2&4\\
4&0\\
\hline
0&3\\
3&3
\end{bmatrix}
\begin{bmatrix}
0&3\\
2&5\\
4&1\\
\hline
1&5\\
4&5
\end{bmatrix}
\begin{bmatrix}
1&0\\
3&2\\
5&4\\
\hline
2&3\\
5&3
\end{bmatrix}
\begin{bmatrix}
1&1\\
3&3\\
5&5\\
\hline
2&5\\
5&5
\end{bmatrix}
\begin{bmatrix}
1&4\\
3&0\\
5&2\\
\hline
1&0\\
4&0
\end{bmatrix}
\begin{bmatrix}
1&5\\
3&1\\
5&3\\
\hline
2&2\\
5&2
\end{bmatrix}
\begin{bmatrix}
0&0\\
0&1\\
2&2\\
2&3\\
4&4\\
4&5\\
\hline
0&1\\
3&1
\end{bmatrix}
\begin{bmatrix}
0&2\\
0&3\\
2&4\\
2&5\\
4&0\\
4&1\\
\hline
1&4\\
4&4
\end{bmatrix}
\begin{bmatrix}
1&0\\
1&1\\
3&2\\
3&3\\
5&4\\
5&5\\
\hline
2&4\\
5&4
\end{bmatrix}
\begin{bmatrix}
1&4\\
1&5\\
3&0\\
3&1\\
5&2\\
5&3\\
\hline
1&1\\
4&1
\end{bmatrix}
\end{aligned}
\]}
\end{lemma}

Рассмотрим $2$-мерные продукционные пары $\aB$ в форме $H$-таблиц с уникальной нижней компонентой. Обозначим здесь это множество таблиц через $\T_{\aB}$. Все $12$ таблиц в нижней компоненте содержат ровно $2$ строки. Среди них $8$ таблиц в верхней компоненте содержат $3$ строки, назовём их $H$-$0235$-таблицами, а $4$ таблицы $6$ строк, назовём их $H$-$14$-таблицами. $H$-$14$-таблицы в крайнем правом столбце верхней компоненты содержит все цифры $\{0,1,2,3,4,5\}$.
Разобьем $H$-$14$-таблицы на строки следующим образом так, чтобы в крайнем правом столбце стояли либо $\{1,3,5\}$, либо $\{0,2,4\}$:
{\scriptsize
\[
\begin{aligned}
\begin{bmatrix}
0&0\\
2&2\\
4&4\\
\hline
0&1\\
3&1
\end{bmatrix}
\begin{bmatrix}
0&1\\
2&3\\
4&5\\
\hline
0&1\\
3&1
\end{bmatrix}
\begin{bmatrix}
0&2\\
2&4\\
4&0\\
\hline
1&4\\
4&4
\end{bmatrix}
\begin{bmatrix}
0&3\\
2&5\\
4&1\\
\hline
1&4\\
4&4
\end{bmatrix}
\begin{bmatrix}
1&0\\
3&2\\
5&4\\
\hline
2&4\\
5&4
\end{bmatrix}
\begin{bmatrix}
1&1\\
3&3\\
5&5\\
\hline
2&4\\
5&4
\end{bmatrix}
\begin{bmatrix}
1&4\\
3&0\\
5&2\\
\hline
1&1\\
4&1
\end{bmatrix}
\begin{bmatrix}
1&5\\
3&1\\
5&3\\
\hline
1&1\\
4&1
\end{bmatrix}
\end{aligned}
\]
}

В такой форме $\aB$ представляется с помощью $16$ таблиц, которые обозначим так же $\T_{\aB}$. Каждая таблица $T\in\T_{\aB}$ имеет свойства, сформулированные в виде следующей леммы. 

\kurglemma
\begin{lemma}[\kurglemmaa]\label{lemma:basictable}
Каждая таблица $T\in\T_{\aB}$ имеет парную ей таблицу в $\T_{\aB}$ с такой же верхней компонентной $\pi_U(T)$. Множество элементов любого столбца в $\pi_U(T)$ равно либо $\{0,2,4\}$, либо $\{1,3,5\}$.
$\pi_U(T)$ имеет ровно три строки и они попарно различаются в каждом столбце. 
Строки нижней компоненты $\pi_O(T)$ отличаются только в крайнем левом столбце. Если множество элементов столбца в $\pi_U(T)$ равно $\{0,2,4\}$, то для элементов $\{x,y\}$ соответствующего столбца в $\pi_O(T)$ выполняется либо  $\{x,y\}\subseteq \{0,3\}$, либо $\{x,y\}\subseteq \{1,4\}$. Если же множество элементов столбца в $\pi_U(T)$ равно $\{1,3,5\}$, то соответственно либо  $\{x,y\}\subseteq \{2,5\}$, либо $\{x,y\}\subseteq \{1,4\}$.
\end{lemma}

Пусть $(n-1)$-мерная $H$-таблица $T$ удовлетворяет лемме~\ref{lemma:basictable}. Тогда множество $\T_{\aB}\circ T$ состоит из $4$ $n$-мерных $H$-таблиц с такими же свойствами. Отсюда следует
\kurgtheorem\begin{theorem}[\kurgtheorema]\label{theorem:basictable}
Множество $H$-таблиц с уникальными нижними компонентами, представляющее множество $\aB^{(n-1)\circ}$, состоит из $3\cdot 4^{n-1}$ элементов. При этом $2\cdot 4^{n-1}$ $H$-таблиц ($H$-$0235$-таблиц) содержат ровно $3$ строки и $4^{n-1}$ таблиц ($H$-$14$-таблиц) ровно $6$ строк в верхних компонентах. $H$-$14$-таблицы в крайнем правом столбце верхней компоненты содержат все символы из алфавита $\Zed_6$. Если $H$-$14$-таблицы раздвоить так, чтобы крайний правый столбец таблиц состоял либо из $\{1,3,5\}$, либо из $\{0,2,4\}$, то все таблицы будут удовлетворять свойствам леммы~\ref{lemma:basictable}.
\end{theorem}

Поскольку до разбиения $H$-$14$-таблиц на две, форма представления продукционных пар была с уникальной нижней компонентной, то верно
\kurgcorollary\begin{corollary}[\kurgcorollarya]
В множестве $\aB^{(n-1)\circ}$ ровно $6\cdot 4^{n-1}$ различных строк.
\end{corollary} 
\kurglemma\begin{lemma}[\kurglemmaa]\label{lemma:B}
Множество $\B$ в форме $H$-таблиц с уникальными нижними компонентами имеет вид:
{\tiny
\[
\begin{aligned}
\begin{bmatrix}
0&0\\
2&0\\
4&0\\
\hline
0&0\\
3&0
\end{bmatrix}
\begin{bmatrix}
0&1\\
2&1\\
4&1\\
\hline
0&2\\
3&2
\end{bmatrix}
\begin{bmatrix}
0&2\\
2&2\\
4&2\\
\hline
0&3\\
3&3
\end{bmatrix}
\begin{bmatrix}
0&3\\
2&3\\
4&3\\
\hline
0&5\\
3&5
\end{bmatrix}
\begin{bmatrix}
0&4\\
2&4\\
4&4\\
\hline
1&0\\
4&0
\end{bmatrix}
\begin{bmatrix}
0&5\\
2&5\\
4&5\\
\hline
1&2\\
4&2
\end{bmatrix}
\begin{bmatrix}
1&0\\
3&0\\
5&0\\
\hline
1&3\\
4&3
\end{bmatrix}
\begin{bmatrix}
1&1\\
3&1\\
5&1\\
\hline
1&5\\
4&5
\end{bmatrix}
\begin{bmatrix}
1&2\\
3&2\\
5&2\\
\hline
2&0\\
5&0
\end{bmatrix}
\begin{bmatrix}
1&3\\
3&3\\
5&3\\
\hline
2&2\\
5&2
\end{bmatrix}
\begin{bmatrix}
1&4\\
3&4\\
5&4\\
\hline
2&3\\
5&3
\end{bmatrix}
\begin{bmatrix}
1&5\\
3&5\\
5&5\\
\hline
2&5\\
5&5
\end{bmatrix}
\begin{bmatrix}
0&0\\
0&1\\
2&0\\
2&1\\
4&0\\
4&1\\
\hline
0&1\\
3&1
\end{bmatrix}
\begin{bmatrix}
0&2\\
0&3\\
2&2\\
2&3\\
4&2\\
4&3\\
\hline
0&4\\
3&4
\end{bmatrix}
\begin{bmatrix}
0&4\\
0&5\\
2&4\\
2&5\\
4&4\\
4&5\\
\hline
1&1\\
4&1
\end{bmatrix}
\begin{bmatrix}
1&0\\
1&1\\
3&0\\
3&1\\
5&0\\
5&1\\
\hline
1&4\\
4&4
\end{bmatrix}
\begin{bmatrix}
1&2\\
1&3\\
3&2\\
3&3\\
5&2\\
5&3\\
\hline
2&1\\
5&1
\end{bmatrix}
\begin{bmatrix}
1&4\\
1&5\\
3&4\\
3&5\\
5&4\\
5&5\\
\hline
2&4\\
5&4
\end{bmatrix}
\end{aligned}
\]
}
\end{lemma}

Рассмотрим $2$-мерные продукционные пары $\B$ в форме $H$-таблиц с уникальной нижней компонентой. Обозначим здесь это множество таблиц через $\T_{\B}$. Все $18$ таблиц в нижней компоненте содержат ровно $2$ строки. $12$ таблиц в верхней компоненте содержат $3$ строки, назовём их $H$-$0235$-таблицами, а $6$ таблиц содержат $6$ строк, назовём их $H$-$14$-таблицами.
Разобьем $H$-$14$-таблицы следующим образом:
{\scriptsize
\[
\begin{aligned}
\begin{bmatrix}
0&0\\
2&0\\
4&0\\
\hline
0&1\\
3&1
\end{bmatrix}
\begin{bmatrix}
0&2\\
2&2\\
4&2\\
\hline
0&4\\
3&4
\end{bmatrix}
\begin{bmatrix}
0&4\\
2&4\\
4&4\\
\hline
1&1\\
4&1
\end{bmatrix}
\begin{bmatrix}
1&0\\
3&0\\
5&0\\
\hline
1&4\\
4&4
\end{bmatrix}
\begin{bmatrix}
1&2\\
3&2\\
5&2\\
\hline
2&1\\
5&1
\end{bmatrix}
\begin{bmatrix}
1&4\\
3&4\\
5&4\\
\hline
2&4\\
5&4
\end{bmatrix}
\begin{bmatrix}
0&1\\
2&1\\
4&1\\
\hline
0&1\\
3&1
\end{bmatrix}
\begin{bmatrix}
0&3\\
2&3\\
4&3\\
\hline
0&4\\
3&4
\end{bmatrix}
\begin{bmatrix}
0&5\\
2&5\\
4&5\\
\hline
1&1\\
4&1
\end{bmatrix}
\begin{bmatrix}
1&1\\
3&1\\
5&1\\
\hline
1&4\\
4&4
\end{bmatrix}
\begin{bmatrix}
1&3\\
3&3\\
5&3\\
\hline
2&1\\
5&1
\end{bmatrix}
\begin{bmatrix}
1&5\\
3&5\\
5&5\\
\hline
2&4\\
5&4
\end{bmatrix}
\end{aligned}
\]
}

В такой форме $\B$ представлено $24$ таблицами, которые обозначим так же $\T_{\B}$.
\kurglemma\begin{lemma}[\kurglemmaa]\label{lemma:basictableB}
Каждая таблица $T\in\T_{\B}$ имеет парную ей таблицу в $\T_{\B}$ с такой же верхней компонентной $\pi_U(T)$. Верхняя компонента
$\pi_U(T)$ имеет ровно три строки. Как строки верхней $\pi_U(T)$ компоненты, так и строки нижней компоненты $\pi_O(T)$ отличаются только в крайнем левом столбце. 
Если элементы столбца в $\pi_U(T)$ принадлежат $\{0,2,4\}$, то для элементов $\{x,y\}$ соответствующего столбца в $\pi_O(T)$ выполняется либо  $\{x,y\}\subseteq \{0,3\}$, либо $\{x,y\}\subseteq \{1,4\}$. Если же элементы столбца в $\pi_U(T)$ принадлежат $\{1,3,5\}$, то либо  $\{x,y\}\subseteq \{2,5\}$, либо $\{x,y\}\subseteq \{1,4\}$.

\end{lemma}

Пусть $(n-1)$-мерная $H$-таблица $T$ удовлетворяет лемме~\ref{lemma:basictableB}. Тогда множество $\T_{\B}\circ T$ состоит из $6$ $n$-мерных $H$-таблиц с такими же свойствами. Отсюда следует
\kurgtheorem\begin{theorem}[\kurgtheorema]\label{theorem:basictableB}
Множество $H$-таблиц с уникальной нижней компонентой, представляющее $n$-мер\-ные про\-дук\-ци\-он\-ные пары $\B^{(n-1)\circ}$, состоит из $3\cdot 6^{n-1}$ элементов. При этом $2\cdot 6^{n-1}$ $H$-таблиц ($H$-$0235$-таблиц) содержат ровно $3$ строки и $6^{n-1}$ таблиц ($H$-$14$-таблиц) ровно $6$ строк в верхних компонентах.  Если $H$-$14$-таблицы раздвоить так, чтобы крайний правый столбец состоял из одинаковых символов, то все таблицы будут удовлетворять свойствам леммы~\ref{lemma:basictableB}.
\end{theorem}

Поскольку до раздвоения $H$-$14$-таблиц форма представления продукционных пар была с уникальной нижней компонентной, то верно
\kurgcorollary\begin{corollary}[\kurgcorollarya]
В множестве $\B^{(n-1)\circ}$ ровно $6^{n}$ различных строк.
\end{corollary}


\kurglemma\begin{lemma}[\kurglemmaa]\label{lemma:gB}
Множество $\gB$ в форме $H$-таблиц с уникальными верхними компонентами имеет вид:
{\scriptsize
\[
\begin{aligned}
\begin{bmatrix}
0&4\\
2&4\\
4&4\\
\hline
0&0\\
0&1\\
3&3\\
3&4
\end{bmatrix}
\begin{bmatrix}
0&5\\
2&5\\
4&5\\
\hline
0&1\\
0&2\\
3&4\\
3&5
\end{bmatrix}
\begin{bmatrix}
0&3\\
2&3\\
4&3\\
\hline
0&4\\
0&5\\
3&1\\
3&2
\end{bmatrix}
\begin{bmatrix}
1&0\\
3&0\\
5&0\\
\hline
2&0\\
2&1\\
5&3\\
5&4
\end{bmatrix}
\begin{bmatrix}
1&1\\
3&1\\
5&1\\
\hline
2&1\\
2&2\\
5&4\\
5&5
\end{bmatrix}
\begin{bmatrix}
1&2\\
3&2\\
5&2\\
\hline
2&3\\
2&4\\
5&0\\
5&1
\end{bmatrix}
\begin{bmatrix}
0&0\\
2&0\\
4&0\\
\hline
0&0\\
0&1\\
1&3\\
1&4\\
3&3\\
3&4\\
4&0\\
4&1
\end{bmatrix}
\begin{bmatrix}
0&1\\
2&1\\
4&1\\
\hline
0&1\\
0&2\\
1&4\\
1&5\\
3&4\\
3&5\\
4&1\\
4&2
\end{bmatrix}
\begin{bmatrix}
0&2\\
2&2\\
4&2\\
\hline
0&3\\
0&4\\
1&0\\
1&1\\
3&0\\
3&1\\
4&3\\
4&4
\end{bmatrix}
\begin{bmatrix}
1&3\\
3&3\\
5&3\\
\hline
1&1\\
1&2\\
2&4\\
2&5\\
4&4\\
4&5\\
5&1\\
5&2
\end{bmatrix}
\begin{bmatrix}
1&4\\
3&4\\
5&4\\
\hline
1&3\\
1&4\\
2&0\\
2&1\\
4&0\\
4&1\\
5&3\\
5&4
\end{bmatrix}
\begin{bmatrix}
1&5\\
3&5\\
5&5\\
\hline
1&4\\
1&5\\
2&1\\
2&2\\
4&1\\
4&2\\
5&4\\
5&5
\end{bmatrix}
\end{aligned}
\]
}
\end{lemma}

\kurglemma\begin{lemma}[\kurglemmaa]\label{lemma:U32} 
Множество $\Univ^{3\times 2}$ в форме $V$-таблиц с уникальными левыми и правыми компонентами имеет вид:
{\scriptsize
\[
\begin{aligned}
\begin{bmatrix}
0&0&2&2&4&4& \vrule & 0&0&0&0&1&1&1&2&2\\
0&0&0&0&0&0& \vrule & 0&0&1&1&1&2&2&3&3\\
0&3&0&3&0&3& \vrule & 0&1&1&2&2&3&4&4&5
\end{bmatrix}
\begin{bmatrix}
0&0&2&2&4&4& \vrule & 2&3&3&3&3\\
0&0&0&0&0&0& \vrule & 4&4&4&5&5\\
1&4&1&4&1&4& \vrule & 0&0&1&1&2
\end{bmatrix}
\\
\begin{bmatrix}
1&1&3&3&5&5& \vrule & 0&0&0&1&1&1&1\\
1&1&1&1&1&1& \vrule & 3&3&4&4&4&5&5\\
2&5&2&5&2&5& \vrule & 1&2&3&3&4&4&5
\end{bmatrix}
\begin{bmatrix}
0&0&2&2&4&4& \vrule & 5&5\\
1&1&1&1&1&1& \vrule & 2&2\\
2&5&2&5&2&5& \vrule & 0&1
\end{bmatrix}
\begin{bmatrix}
0&0&2&2&4&4& \vrule & 4&4&4&4&5\\
1&1&1&1&1&1& \vrule & 0&0&1&1&1\\
1&4&1&4&1&4& \vrule & 3&4&4&5&5
\end{bmatrix}
\\
\begin{bmatrix}
1&1&3&3&5&5& \vrule & 2&2&2&2&3&3&3&4&4\\
2&2&2&2&2&2& \vrule & 0&0&1&1&1&2&2&3&3\\
0&3&0&3&0&3& \vrule & 0&1&1&2&2&3&4&4&5
\end{bmatrix}
\begin{bmatrix}
1&1&3&3&5&5& \vrule & 4&5&5&5&5\\
2&2&2&2&2&2& \vrule & 4&4&4&5&5\\
1&4&1&4&1&4& \vrule & 0&0&1&1&2
\end{bmatrix}
\\
\begin{bmatrix}
0&0&2&2&4&4& \vrule & 1&1&2&2&2&3&3&3&3\\
3&3&3&3&3&3& \vrule & 2&2&3&3&4&4&4&5&5\\
2&5&2&5&2&5& \vrule & 0&1&1&2&3&3&4&4&5
\end{bmatrix}
\begin{bmatrix}
0&0&2&2&4&4& \vrule & 0&0&0&0&1\\
3&3&3&3&3&3& \vrule & 0&0&1&1&1\\
1&4&1&4&1&4& \vrule & 3&4&4&5&5
\end{bmatrix}
\\
\begin{bmatrix}
0&0&2&2&4&4& \vrule & 4&4&4&4&5&5&5\\
4&4&4&4&4&4& \vrule & 0&0&1&1&1&2&2\\
0&3&0&3&0&3& \vrule & 0&1&1&2&2&3&4
\end{bmatrix}
\begin{bmatrix}
1&1&3&3&5&5& \vrule & 0&0\\
4&4&4&4&4&4& \vrule & 3&3\\
0&3&0&3&0&3& \vrule & 4&5
\end{bmatrix}
\begin{bmatrix}
1&1&3&3&5&5& \vrule & 0&1&1&1&1\\
4&4&4&4&4&4& \vrule & 4&4&4&5&5\\
1&4&1&4&1&4& \vrule & 0&0&1&1&2
\end{bmatrix}
\\
\begin{bmatrix}
1&1&3&3&5&5& \vrule & 3&3&4&4&4&5&5&5&5\\
5&5&5&5&5&5& \vrule & 2&2&3&3&4&4&4&5&5\\
2&5&2&5&2&5& \vrule & 0&1&1&2&3&3&4&4&5
\end{bmatrix}
\begin{bmatrix}
1&1&3&3&5&5& \vrule & 2&2&2&2&3\\
5&5&5&5&5&5& \vrule & 0&0&1&1&1\\
1&4&1&4&1&4& \vrule & 3&4&4&5&5
\end{bmatrix}
\end{aligned}
\]
}
\end{lemma}

\kurglemma\begin{lemma}[\kurglemmaa]\label{lemma:U23}
Множество $\Univ^{2\times 3}$ имеет вид:
{\scriptsize
\[
\begin{aligned}
\begin{bmatrix}
0&0&0\\
2&0&0\\
4&0&0\\
\hline
0&0&0\\
3&0&0
\end{bmatrix}
\begin{bmatrix}
0&0&1\\
2&0&1\\
4&0&1\\
\hline
0&0&2\\
3&0&2
\end{bmatrix}
\begin{bmatrix}
0&0&2\\
2&0&2\\
4&0&2\\
\hline
0&0&3\\
3&0&3
\end{bmatrix}
\begin{bmatrix}
0&0&3\\
2&0&3\\
4&0&3\\
\hline
0&0&5\\
3&0&5
\end{bmatrix}
\begin{bmatrix}
0&0&4\\
2&0&4\\
4&0&4\\
\hline
0&1&0\\
3&1&0
\end{bmatrix}
\begin{bmatrix}
0&0&5\\
2&0&5\\
4&0&5\\
\hline
0&1&2\\
3&1&2
\end{bmatrix}
\begin{bmatrix}
0&1&0\\
2&1&0\\
4&1&0\\
\hline
0&1&3\\
3&1&3
\end{bmatrix}
\begin{bmatrix}
0&1&1\\
2&1&1\\
4&1&1\\
\hline
0&1&5\\
3&1&5
\end{bmatrix}
\begin{bmatrix}
0&1&2\\
2&1&2\\
4&1&2\\
\hline
0&2&0\\
3&2&0
\end{bmatrix}
\begin{bmatrix}
0&1&3\\
2&1&3\\
4&1&3\\
\hline
0&2&2\\
3&2&2
\end{bmatrix}
\begin{bmatrix}
0&1&4\\
2&1&4\\
4&1&4\\
\hline
0&2&3\\
3&2&3
\end{bmatrix}
\begin{bmatrix}
0&1&5\\
2&1&5\\
4&1&5\\
\hline
0&2&5\\
3&2&5
\end{bmatrix}
\\ 
\begin{bmatrix}
0&2&0\\
2&2&0\\
4&2&0\\
\hline
0&3&0\\
3&3&0
\end{bmatrix}
\begin{bmatrix}
0&2&1\\
2&2&1\\
4&2&1\\
\hline
0&3&2\\
3&3&2
\end{bmatrix}
\begin{bmatrix}
0&2&2\\
2&2&2\\
4&2&2\\
\hline
0&3&3\\
3&3&3
\end{bmatrix}
\begin{bmatrix}
0&2&3\\
2&2&3\\
4&2&3\\
\hline
0&3&5\\
3&3&5
\end{bmatrix}
\begin{bmatrix}
0&2&4\\
2&2&4\\
4&2&4\\
\hline
0&4&0\\
3&4&0
\end{bmatrix}
\begin{bmatrix}
0&2&5\\
2&2&5\\
4&2&5\\
\hline
0&4&2\\
3&4&2
\end{bmatrix}
\begin{bmatrix}
0&3&0\\
2&3&0\\
4&3&0\\
\hline
0&4&3\\
3&4&3
\end{bmatrix}
\begin{bmatrix}
0&3&1\\
2&3&1\\
4&3&1\\
\hline
0&4&5\\
3&4&5
\end{bmatrix}
\begin{bmatrix}
0&3&2\\
2&3&2\\
4&3&2\\
\hline
0&5&0\\
3&5&0
\end{bmatrix}
\begin{bmatrix}
0&3&3\\
2&3&3\\
4&3&3\\
\hline
0&5&2\\
3&5&2
\end{bmatrix}
\begin{bmatrix}
0&3&4\\
2&3&4\\
4&3&4\\
\hline
0&5&3\\
3&5&3
\end{bmatrix}
\begin{bmatrix}
0&3&5\\
2&3&5\\
4&3&5\\
\hline
0&5&5\\
3&5&5
\end{bmatrix}
\\ 
\begin{bmatrix}
0&4&0\\
2&4&0\\
4&4&0\\
\hline
1&0&0\\
4&0&0
\end{bmatrix}
\begin{bmatrix}
0&4&1\\
2&4&1\\
4&4&1\\
\hline
1&0&2\\
4&0&2
\end{bmatrix}
\begin{bmatrix}
0&4&2\\
2&4&2\\
4&4&2\\
\hline
1&0&3\\
4&0&3
\end{bmatrix}
\begin{bmatrix}
0&4&3\\
2&4&3\\
4&4&3\\
\hline
1&0&5\\
4&0&5
\end{bmatrix}
\begin{bmatrix}
0&4&4\\
2&4&4\\
4&4&4\\
\hline
1&1&0\\
4&1&0
\end{bmatrix}
\begin{bmatrix}
0&4&5\\
2&4&5\\
4&4&5\\
\hline
1&1&2\\
4&1&2
\end{bmatrix}
\begin{bmatrix}
0&5&0\\
2&5&0\\
4&5&0\\
\hline
1&1&3\\
4&1&3
\end{bmatrix}
\begin{bmatrix}
0&5&1\\
2&5&1\\
4&5&1\\
\hline
1&1&5\\
4&1&5
\end{bmatrix}
\begin{bmatrix}
0&5&2\\
2&5&2\\
4&5&2\\
\hline
1&2&0\\
4&2&0
\end{bmatrix}
\begin{bmatrix}
0&5&3\\
2&5&3\\
4&5&3\\
\hline
1&2&2\\
4&2&2
\end{bmatrix}
\begin{bmatrix}
0&5&4\\
2&5&4\\
4&5&4\\
\hline
1&2&3\\
4&2&3
\end{bmatrix}
\begin{bmatrix}
0&5&5\\
2&5&5\\
4&5&5\\
\hline
1&2&5\\
4&2&5
\end{bmatrix}
\\ 
\begin{bmatrix}
1&0&0\\
3&0&0\\
5&0&0\\
\hline
1&3&0\\
4&3&0
\end{bmatrix}
\begin{bmatrix}
1&0&1\\
3&0&1\\
5&0&1\\
\hline
1&3&2\\
4&3&2
\end{bmatrix}
\begin{bmatrix}
1&0&2\\
3&0&2\\
5&0&2\\
\hline
1&3&3\\
4&3&3
\end{bmatrix}
\begin{bmatrix}
1&0&3\\
3&0&3\\
5&0&3\\
\hline
1&3&5\\
4&3&5
\end{bmatrix}
\begin{bmatrix}
1&0&4\\
3&0&4\\
5&0&4\\
\hline
1&4&0\\
4&4&0
\end{bmatrix}
\begin{bmatrix}
1&0&5\\
3&0&5\\
5&0&5\\
\hline
1&4&2\\
4&4&2
\end{bmatrix}
\begin{bmatrix}
1&1&0\\
3&1&0\\
5&1&0\\
\hline
1&4&3\\
4&4&3
\end{bmatrix}
\begin{bmatrix}
1&1&1\\
3&1&1\\
5&1&1\\
\hline
1&4&5\\
4&4&5
\end{bmatrix}
\begin{bmatrix}
1&1&2\\
3&1&2\\
5&1&2\\
\hline
1&5&0\\
4&5&0
\end{bmatrix}
\begin{bmatrix}
1&1&3\\
3&1&3\\
5&1&3\\
\hline
1&5&2\\
4&5&2
\end{bmatrix}
\begin{bmatrix}
1&1&4\\
3&1&4\\
5&1&4\\
\hline
1&5&3\\
4&5&3
\end{bmatrix}
\begin{bmatrix}
1&1&5\\
3&1&5\\
5&1&5\\
\hline
1&5&5\\
4&5&5
\end{bmatrix}
\\ 
\begin{bmatrix}
1&2&0\\
3&2&0\\
5&2&0\\
\hline
2&0&0\\
5&0&0
\end{bmatrix}
\begin{bmatrix}
1&2&1\\
3&2&1\\
5&2&1\\
\hline
2&0&2\\
5&0&2
\end{bmatrix}
\begin{bmatrix}
1&2&2\\
3&2&2\\
5&2&2\\
\hline
2&0&3\\
5&0&3
\end{bmatrix}
\begin{bmatrix}
1&2&3\\
3&2&3\\
5&2&3\\
\hline
2&0&5\\
5&0&5
\end{bmatrix}
\begin{bmatrix}
1&2&4\\
3&2&4\\
5&2&4\\
\hline
2&1&0\\
5&1&0
\end{bmatrix}
\begin{bmatrix}
1&2&5\\
3&2&5\\
5&2&5\\
\hline
2&1&2\\
5&1&2
\end{bmatrix}
\begin{bmatrix}
1&3&0\\
3&3&0\\
5&3&0\\
\hline
2&1&3\\
5&1&3
\end{bmatrix}
\begin{bmatrix}
1&3&1\\
3&3&1\\
5&3&1\\
\hline
2&1&5\\
5&1&5
\end{bmatrix}
\begin{bmatrix}
1&3&2\\
3&3&2\\
5&3&2\\
\hline
2&2&0\\
5&2&0
\end{bmatrix}
\begin{bmatrix}
1&3&3\\
3&3&3\\
5&3&3\\
\hline
2&2&2\\
5&2&2
\end{bmatrix}
\begin{bmatrix}
1&3&4\\
3&3&4\\
5&3&4\\
\hline
2&2&3\\
5&2&3
\end{bmatrix}
\begin{bmatrix}
1&3&5\\
3&3&5\\
5&3&5\\
\hline
2&2&5\\
5&2&5
\end{bmatrix}
\\ 
\begin{bmatrix}
1&4&0\\
3&4&0\\
5&4&0\\
\hline
2&3&0\\
5&3&0
\end{bmatrix}
\begin{bmatrix}
1&4&1\\
3&4&1\\
5&4&1\\
\hline
2&3&2\\
5&3&2
\end{bmatrix}
\begin{bmatrix}
1&4&2\\
3&4&2\\
5&4&2\\
\hline
2&3&3\\
5&3&3
\end{bmatrix}
\begin{bmatrix}
1&4&3\\
3&4&3\\
5&4&3\\
\hline
2&3&5\\
5&3&5
\end{bmatrix}
\begin{bmatrix}
1&4&4\\
3&4&4\\
5&4&4\\
\hline
2&4&0\\
5&4&0
\end{bmatrix}
\begin{bmatrix}
1&4&5\\
3&4&5\\
5&4&5\\
\hline
2&4&2\\
5&4&2
\end{bmatrix}
\begin{bmatrix}
1&5&0\\
3&5&0\\
5&5&0\\
\hline
2&4&3\\
5&4&3
\end{bmatrix}
\begin{bmatrix}
1&5&1\\
3&5&1\\
5&5&1\\
\hline
2&4&5\\
5&4&5
\end{bmatrix}
\begin{bmatrix}
1&5&2\\
3&5&2\\
5&5&2\\
\hline
2&5&0\\
5&5&0
\end{bmatrix}
\begin{bmatrix}
1&5&3\\
3&5&3\\
5&5&3\\
\hline
2&5&2\\
5&5&2
\end{bmatrix}
\begin{bmatrix}
1&5&4\\
3&5&4\\
5&5&4\\
\hline
2&5&3\\
5&5&3
\end{bmatrix}
\begin{bmatrix}
1&5&5\\
3&5&5\\
5&5&5\\
\hline
2&5&5\\
5&5&5
\end{bmatrix}
\\ 
\begin{bmatrix}
0&0&0\\
0&0&1\\
2&0&0\\
2&0&1\\
4&0&0\\
4&0&1\\
\hline
0&0&1\\
3&0&1
\end{bmatrix}
\begin{bmatrix}
0&0&2\\
0&0&3\\
2&0&2\\
2&0&3\\
4&0&2\\
4&0&3\\
\hline
0&0&4\\
3&0&4
\end{bmatrix}
\begin{bmatrix}
0&0&4\\
0&0&5\\
2&0&4\\
2&0&5\\
4&0&4\\
4&0&5\\
\hline
0&1&1\\
3&1&1
\end{bmatrix}
\begin{bmatrix}
0&1&0\\
0&1&1\\
2&1&0\\
2&1&1\\
4&1&0\\
4&1&1\\
\hline
0&1&4\\
3&1&4
\end{bmatrix}
\begin{bmatrix}
0&1&2\\
0&1&3\\
2&1&2\\
2&1&3\\
4&1&2\\
4&1&3\\
\hline
0&2&1\\
3&2&1
\end{bmatrix}
\begin{bmatrix}
0&1&4\\
0&1&5\\
2&1&4\\
2&1&5\\
4&1&4\\
4&1&5\\
\hline
0&2&4\\
3&2&4
\end{bmatrix}
\begin{bmatrix}
0&2&0\\
0&2&1\\
2&2&0\\
2&2&1\\
4&2&0\\
4&2&1\\
\hline
0&3&1\\
3&3&1
\end{bmatrix}
\begin{bmatrix}
0&2&2\\
0&2&3\\
2&2&2\\
2&2&3\\
4&2&2\\
4&2&3\\
\hline
0&3&4\\
3&3&4
\end{bmatrix}
\begin{bmatrix}
0&2&4\\
0&2&5\\
2&2&4\\
2&2&5\\
4&2&4\\
4&2&5\\
\hline
0&4&1\\
3&4&1
\end{bmatrix}
\begin{bmatrix}
0&3&0\\
0&3&1\\
2&3&0\\
2&3&1\\
4&3&0\\
4&3&1\\
\hline
0&4&4\\
3&4&4
\end{bmatrix}
\begin{bmatrix}
0&3&2\\
0&3&3\\
2&3&2\\
2&3&3\\
4&3&2\\
4&3&3\\
\hline
0&5&1\\
3&5&1
\end{bmatrix}
\begin{bmatrix}
0&3&4\\
0&3&5\\
2&3&4\\
2&3&5\\
4&3&4\\
4&3&5\\
\hline
0&5&4\\
3&5&4
\end{bmatrix}
\\ 
\begin{bmatrix}
0&4&0\\
0&4&1\\
2&4&0\\
2&4&1\\
4&4&0\\
4&4&1\\
\hline
1&0&1\\
4&0&1
\end{bmatrix}
\begin{bmatrix}
0&4&2\\
0&4&3\\
2&4&2\\
2&4&3\\
4&4&2\\
4&4&3\\
\hline
1&0&4\\
4&0&4
\end{bmatrix}
\begin{bmatrix}
0&4&4\\
0&4&5\\
2&4&4\\
2&4&5\\
4&4&4\\
4&4&5\\
\hline
1&1&1\\
4&1&1
\end{bmatrix}
\begin{bmatrix}
0&5&0\\
0&5&1\\
2&5&0\\
2&5&1\\
4&5&0\\
4&5&1\\
\hline
1&1&4\\
4&1&4
\end{bmatrix}
\begin{bmatrix}
0&5&2\\
0&5&3\\
2&5&2\\
2&5&3\\
4&5&2\\
4&5&3\\
\hline
1&2&1\\
4&2&1
\end{bmatrix}
\begin{bmatrix}
0&5&4\\
0&5&5\\
2&5&4\\
2&5&5\\
4&5&4\\
4&5&5\\
\hline
1&2&4\\
4&2&4
\end{bmatrix}
\begin{bmatrix}
1&0&0\\
1&0&1\\
3&0&0\\
3&0&1\\
5&0&0\\
5&0&1\\
\hline
1&3&1\\
4&3&1
\end{bmatrix}
\begin{bmatrix}
1&0&2\\
1&0&3\\
3&0&2\\
3&0&3\\
5&0&2\\
5&0&3\\
\hline
1&3&4\\
4&3&4
\end{bmatrix}
\begin{bmatrix}
1&0&4\\
1&0&5\\
3&0&4\\
3&0&5\\
5&0&4\\
5&0&5\\
\hline
1&4&1\\
4&4&1
\end{bmatrix}
\begin{bmatrix}
1&1&0\\
1&1&1\\
3&1&0\\
3&1&1\\
5&1&0\\
5&1&1\\
\hline
1&4&4\\
4&4&4
\end{bmatrix}
\begin{bmatrix}
1&1&2\\
1&1&3\\
3&1&2\\
3&1&3\\
5&1&2\\
5&1&3\\
\hline
1&5&1\\
4&5&1
\end{bmatrix}
\begin{bmatrix}
1&1&4\\
1&1&5\\
3&1&4\\
3&1&5\\
5&1&4\\
5&1&5\\
\hline
1&5&4\\
4&5&4
\end{bmatrix}
\\ 
\begin{bmatrix}
1&2&0\\
1&2&1\\
3&2&0\\
3&2&1\\
5&2&0\\
5&2&1\\
\hline
2&0&1\\
5&0&1
\end{bmatrix}
\begin{bmatrix}
1&2&2\\
1&2&3\\
3&2&2\\
3&2&3\\
5&2&2\\
5&2&3\\
\hline
2&0&4\\
5&0&4
\end{bmatrix}
\begin{bmatrix}
1&2&4\\
1&2&5\\
3&2&4\\
3&2&5\\
5&2&4\\
5&2&5\\
\hline
2&1&1\\
5&1&1
\end{bmatrix}
\begin{bmatrix}
1&3&0\\
1&3&1\\
3&3&0\\
3&3&1\\
5&3&0\\
5&3&1\\
\hline
2&1&4\\
5&1&4
\end{bmatrix}
\begin{bmatrix}
1&3&2\\
1&3&3\\
3&3&2\\
3&3&3\\
5&3&2\\
5&3&3\\
\hline
2&2&1\\
5&2&1
\end{bmatrix}
\begin{bmatrix}
1&3&4\\
1&3&5\\
3&3&4\\
3&3&5\\
5&3&4\\
5&3&5\\
\hline
2&2&4\\
5&2&4
\end{bmatrix}
\begin{bmatrix}
1&4&0\\
1&4&1\\
3&4&0\\
3&4&1\\
5&4&0\\
5&4&1\\
\hline
2&3&1\\
5&3&1
\end{bmatrix}
\begin{bmatrix}
1&4&2\\
1&4&3\\
3&4&2\\
3&4&3\\
5&4&2\\
5&4&3\\
\hline
2&3&4\\
5&3&4
\end{bmatrix}
\begin{bmatrix}
1&4&4\\
1&4&5\\
3&4&4\\
3&4&5\\
5&4&4\\
5&4&5\\
\hline
2&4&1\\
5&4&1
\end{bmatrix}
\begin{bmatrix}
1&5&0\\
1&5&1\\
3&5&0\\
3&5&1\\
5&5&0\\
5&5&1\\
\hline
2&4&4\\
5&4&4
\end{bmatrix}
\begin{bmatrix}
1&5&2\\
1&5&3\\
3&5&2\\
3&5&3\\
5&5&2\\
5&5&3\\
\hline
2&5&1\\
5&5&1
\end{bmatrix}
\begin{bmatrix}
1&5&4\\
1&5&5\\
3&5&4\\
3&5&5\\
5&5&4\\
5&5&5\\
\hline
2&5&4\\
5&5&4
\end{bmatrix}
\end{aligned}
\]
}
\end{lemma}


\begin{figure}
\centering
\includegraphics[width=0.4\textwidth]{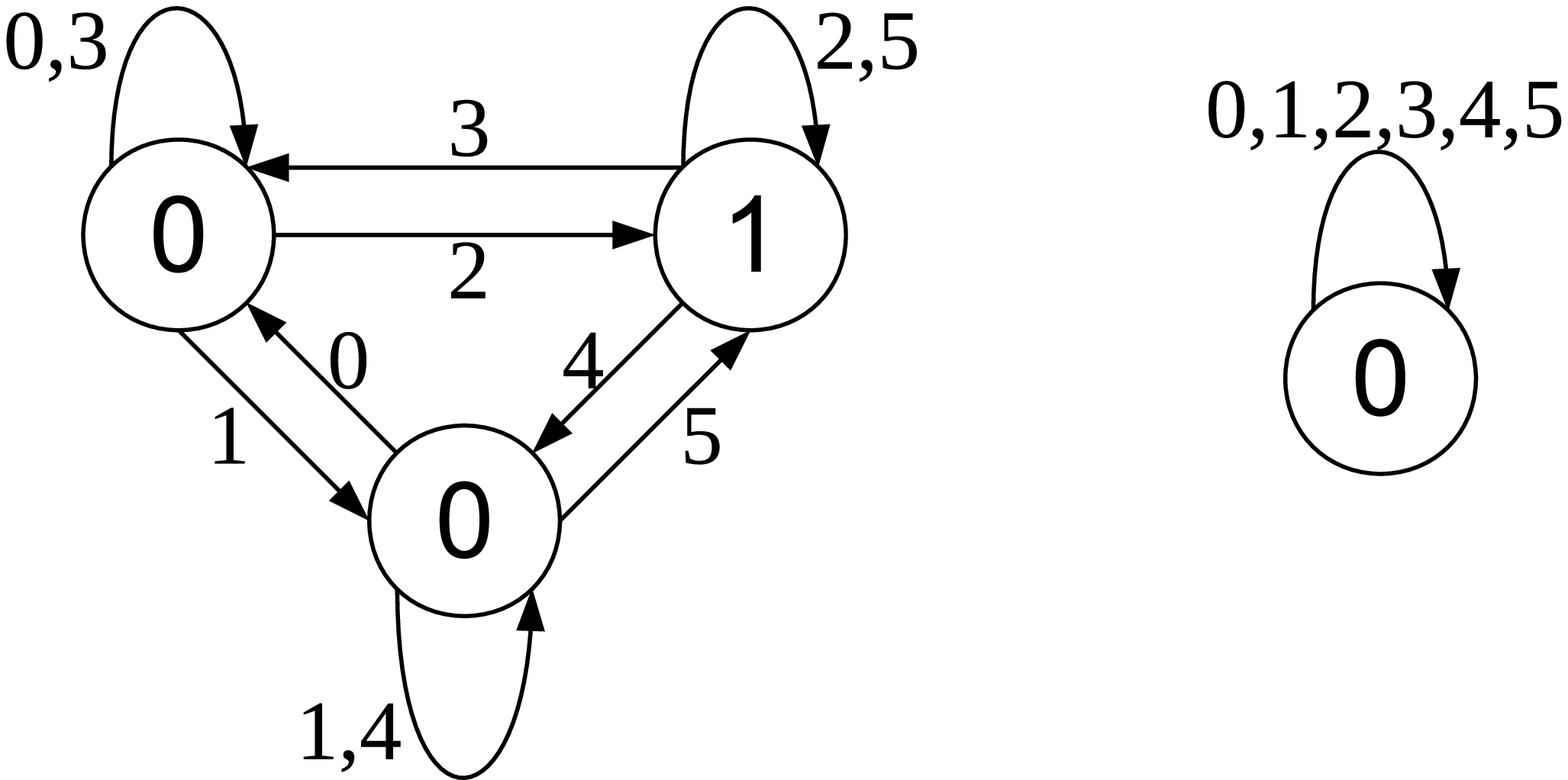}
\caption{Автомат  (слева), порождающий язык верхних и язык нижних строк (читаемых справа налево) слов из $\aB^{n\circ}$, $n\in\Zed^{+}$. Автомат  (справа), порождающий язык верхних и язык нижних строк слов из $\B^{n\circ}$,  $n\in\Zed^{+}$. Он же для языка верхних и языка нижних строк слов из $\gB^{n\circ}$,  $n\in\Zed^{+}$. Все состояния начальные и заключительные. Пустые слова не принимаются во внимание.}
\label{AutomatonVHTuplesDiag22.eps}
\end{figure}

\kurglemma\begin{lemma}[\kurglemmaa]\label{lemma:duality} (о симметрии)
Множество $\B$, а следовательно и $\aB$, $\gB$, $\Univ$, инвариантно относительно преобразования $f(x)=5-x$ алфавита $\Zed_6$.
\end{lemma}

\begin{definition}\label{definition:duallanguages}
Языки $W$ и $W^{dual}$ в алфавите $\Zed_6$ назовём двойственными, если один из другого получается преобразованием $f(x)=5-x$ алфавита.
\end{definition}

\kurgcorollary\begin{corollary}[\kurgcorollarya]
Верно, что $\B=\B^{dual}$, $\aB=\aB^{dual}$, $\gB=\gB^{dual}$, $\Univ=\Univ^{dual}$.
\end{corollary}


\kurglemma\begin{lemma}[\kurglemmaa]\label{lemma:leftableitung}\label{lemma:upableitung}\label{lemma:downableitung}
Верно следующее:
\begin{enumerate}
	\item\label{lemma:leftableitung1}
	Если
$\left(
\begin{smallmatrix}
a_1 &  a\\
b_1 & b\\
c_1 & c
\end{smallmatrix}\right)$,
$\left(\begin{smallmatrix}
a_2 &  a\\
b_2 & b\\
c_2 & c
\end{smallmatrix}\right)\in \Univ^{3\times 2}$,
то $b_1 = b_2$.
	\item Если
$\left(\begin{smallmatrix}
a_1 &  b_1 & c_1\\
a & b & c
\end{smallmatrix}\right)$,
$\left(\begin{smallmatrix}
a_2 &  b_2 & c_2\\
a & b & c
\end{smallmatrix}\right) \in \Univ^{2\times 3}$,
то $b_1 = b_2$. 
	\item Если
$\left(\begin{smallmatrix}
a & b & c \\
a_1 &  b_1 & c_1
\end{smallmatrix}\right)$, 
$\left(\begin{smallmatrix}
a & b & c \\
a_2 &  b_2 & c_2
\end{smallmatrix}\right)
\in \Univ^{2\times 3}$, то $b_1 = b_2$.
\end{enumerate}
\end{lemma}

\begin{proof}
Следует из леммы~\ref{lemma:U32} и теоремы~\ref{theorem:basictableB}.\qed
\end{proof}
\kurglemma\begin{lemma}[\kurglemmaa]\label{lemma:BaBgB}
Верно, что 
$
(\B^{\pm\omega\circ})^{\pm\omega\bullet}\sim
(\aB^{\pm\omega\circ})^{\pm\omega\bullet}\sim
(\gB^{\pm\omega\circ})^{\pm\omega\bullet}
$.
\end{lemma}
\begin{proof}
Для того, чтобы доказать $(\B^{\pm\omega\circ})^{\pm\omega\bullet}\sim
(\aB^{\pm\omega\circ})^{\pm\omega\bullet}$, достаточно показать, что из 
$\left(\begin{smallmatrix}
a&b\\
c&d
\end{smallmatrix}\right)
$,
$\left(\begin{smallmatrix}
c&d\\
e&f
\end{smallmatrix}\right)\in\B
$,  
$
\left(\begin{smallmatrix}
c&d\\
e&f
\end{smallmatrix}\right)
\bullet
\left(\begin{smallmatrix}
a&b\\
c&d
\end{smallmatrix}\right)
=
\left(\begin{smallmatrix}
a&b\\
c&d\\
e&f
\end{smallmatrix}\right)
$
следует 
$\left(\begin{smallmatrix}
c&b\\
e&d
\end{smallmatrix}\right)\in\aB
$, и, наоборот, из 
$\left(\begin{smallmatrix}
a&b\\
c&d
\end{smallmatrix}\right)
$,
$\left(\begin{smallmatrix}
c&d\\
e&f
\end{smallmatrix}\right)\in\aB
$, 
$
\left(\begin{smallmatrix}
c&d\\
e&f
\end{smallmatrix}\right)
\bullet
\left(\begin{smallmatrix}
a&b\\
c&d
\end{smallmatrix}\right)
=
\left(\begin{smallmatrix}
a&b\\
c&d\\
e&f
\end{smallmatrix}\right)
$
следует 
$\left(\begin{smallmatrix}
a&d\\
c&f
\end{smallmatrix}\right)\in\B
$.
В силу конечности рассматриваемых объектов эти факты устанавливаются экспериментально.
Отношение 
$
(\aB^{\pm\omega\circ})^{\pm\omega\bullet}\sim
(\gB^{\pm\omega\circ})^{\pm\omega\bullet}
$ доказывается аналогично.
\end{proof}
\kurglemma\begin{lemma}[\kurglemmaa]\label{lemma:beztupikov}
Множества $\B^{n\circ}$, $\aB^{n\circ}$, $\gB^{n\circ}$, $\B^{\pm\omega\circ}$, $\aB^{\pm\omega\circ}$, $\gB^{\pm\omega\circ}$ не содержат вертикальных тупиков, $n\in \Zed^{+}$.
\end{lemma}
\begin{proof}
По конечным множествам $\B$, $\aB$, $\gB$ построим автоматы, допускающие языки нижних строк и языки верхних строк слов соответственно из  $\bigcup_{n\in\Zed^{+}}{\B^{n\circ}}$, $\bigcup_{n\in\Zed^{+}}{\aB^{n\circ}}$ и $\bigcup_{n\in\Zed^{+}}{\gB^{n\circ}}$. Автоматы показаны на рисунке~\ref{AutomatonVHTuplesDiag22.eps}. Из построений следует, что языки верхних и языки нижних строк совпадают для каждого рассматриваемого множества.
\qed\end{proof}

\kurglemma\begin{lemma}[\kurglemmaa]\label{lemma:U23=BB}
Верны равенства:
{\small
\[
\Univ^{2\times 3}=\B\circ\B,\quad
\aUniv^{2\times 3}=\aB\circ\aB,\quad
\gUniv^{2\times 3}=\gB\circ\gB.
\]
\[
\Univ[0:2,:] = \B^{\pm\omega\circ},\quad
\aUniv[0:2,:] = \aB^{\pm\omega\circ},\quad
\gUniv[0:2,:] = \gB^{\pm\omega\circ}.
\]}
\end{lemma}
\begin{proof}
Первые три равенства доказываются экспериментально.

Включение $\Univ[0:2,:] \subseteq \B^{\pm\omega\circ}$ верно по определению.
Обратное включение верно в силу равенства $\Univ^{2\times 3}=\B\circ\B$, так как если 
$
\left(\begin{smallmatrix}
a & b & c \\
a_1 &  b_1 & c_1
\end{smallmatrix}\right)\in \Univ^{2\times 3}
$,
то $b_1$ формально получается умножением $abc$ на $3/2$, а значит вся нижняя строка $\B^{\pm\omega\circ}$ получается из верхней умножением на $3/2$. Далее отсюда следует, что $\Univ=(\B^{\pm\omega\circ})^{\pm\omega\bullet}$. Теперь оставшиеся равенства
$
\aUniv[0:2,:] = \aB^{\pm\omega\circ}
$,
$
\gUniv[0:2,:] = \gB^{\pm\omega\circ}
$
следуют из
$
(\B^{\pm\omega\circ})^{\pm\omega\bullet}\sim
(\aB^{\pm\omega\circ})^{\pm\omega\bullet}\sim
(\gB^{\pm\omega\circ})^{\pm\omega\bullet}
$ и того, что $\aB^{\pm\omega\circ}$, $\gB^{\pm\omega\circ}$ не содержат вертикальных тупиков.
\qed\end{proof}
\kurgcorollary\begin{corollary}[\kurgcorollarya]
Справедливы равенства: 
\[
\Univ^{2\times n} = \B^{(n-1)\circ}, \quad
\Univ=(\B^{\pm\omega\circ})^{\pm\omega\bullet}.
\]
\[
\aUniv^{2\times n} = \aB^{(n-1)\circ}, \quad
\aUniv=(\aB^{\pm\omega\circ})^{\pm\omega\bullet},
\]
\[
\gUniv^{2\times n} = \gB^{(n-1)\circ}, \quad
\gUniv=(\gB^{\pm\omega\circ})^{\pm\omega\bullet}.
\]
\end{corollary}

Заметим, что $\Univ^{n\times 2} \neq \B^{(n-1)\bullet}$, $\aUniv^{n\times 2} \neq \aB^{(n-1)\bullet}$, $\gUniv^{n\times 2} \neq \gB^{(n-1)\bullet}$. 

Через $\La$ обозначим частичное отображение на словах такое, что 
$\La(x_1x_2x_3) = b$ тогда и только тогда, когда 
\[
\begin{pmatrix}
a &  x_1\\
b & x_2\\
c & x_3
\end{pmatrix}\in \Univ^{3\times 2},
\]
а для $n\ge 3$
$
\La(x_1x_2\ldots x_n)=\La(x_1x_2x_3)\La(x_2x_3x_4)\ldots \La(x_{n-2}x_{n-1}x_n)
$.
И, наконец, для бесконечных слов: пусть $c:\Zed\rightarrow \Zed_6$, $c':\Zed\rightarrow \Zed_6$, тогда $\La(c)=c'$, если для любого $i\in \Zed$ $\La(c(i-1)c(i)c(i+1)) = c'(i)$.

\kurgcorollary\begin{corollary}[\kurgcorollarya]
Отображение $\La$ является функцией.
\end{corollary}

\begin{definition}
Слово $w$ назовём правильным, если определена функция $\La(w)$. При этом, если $w=abc\in\Zed_6^{3}$ правильное слово, то 
правильными также назовём слова $ab$, $bc$, $a$, $b$, $c$ и пустое слово.
\end{definition}

\kurgtheorem\begin{theorem}[\kurgtheorema]\label{theorem:pravilnyeslova}
Если слово $w$ правильное, то и $\La(w)$ правильное.
\end{theorem}
\begin{proof}
Утверждение достаточно доказать для всех слов $w$ длины $5$. Поскольку их конечное число, то лемма устанавливается экспериментально.
\qed\end{proof}

%

\kurglemma\begin{lemma}[\kurglemmaa]\label{lemma:pravilny}
Для любого $n>1$ правые столбцы слов из $\B^{n\bullet}$ правильные.
\end{lemma}
\begin{proof}
Утверждение достаточно доказать для $n=2$. Этот факт устанавливается экспериментально.
\qed\end{proof}

\kurglemma\begin{lemma}[\kurglemmaa]
Пусть $W=\B^{(n-1)\bullet}$, $n>1$.  Множество крайних левых столбцов слов из $W$ содержится во множестве крайних правых столбцов слов из $W$.
\end{lemma}
\begin{proof}
Рассмотрим произвольное слово 
\[
\begin{pmatrix}
b_{-1}&c_{-1}\\
b_0&c_0\\
b_1&c_1\\
\vdots&\vdots\\
b_{n+2}&c_{n+2}
\end{pmatrix} \in \B^{(n+3)\bullet}.
\]
Это значит также, что 
\[
\begin{pmatrix}
b_1&c_1\\
\vdots&\vdots\\
b_{n}&c_{n}
\end{pmatrix}\in \B^{(n-1)\bullet}.
\]
По лемме~\ref{lemma:pravilny} и теореме~\ref{theorem:pravilnyeslova} слово $b_0b_1\ldots b_{n+1}$ правильное. 

Пусть $a_1a_2\ldots a_{n}=\La(b_0b_1\ldots b_{n+1})$.
Тогда 
\[
\begin{pmatrix}
a_1&b_1\\
a_2&b_2\\
\vdots&\vdots\\
a_{n}&b_{n}
\end{pmatrix}\in\B^{(n-1)\bullet}.
\]
Что и требовалось доказать.
\qed\end{proof}
\kurgcorollary\begin{corollary}[\kurgcorollarya]
Множество всех горизонтально нетупиковых слов из $\B^{(n-1)\bullet}$ равно $\Univ^{n\times 2}$.
\end{corollary}
\begin{proof}
Понятно, что $\Univ^{n\times 2}\subseteq W$.
Пусть $w\in W^{\pm\omega\circ}$. В силу построений каждая строка в $w$ получается из предыдущей умножением на $3/2$. Следовательно $W\subseteq \Univ^{n\times 2}$. Что и требовалось доказать.
\qed\end{proof}

\kurglemma\begin{lemma}[\kurglemmaa]
Автоматы, порождающие правильные слова в алфавите $\{0,2,3,5\}$ и в алфавите $\Zed_6$ имеют вид,
как показано на рисунке~\ref{I_0235.eps} слева и справа соответственно. 
\end{lemma}
\begin{proof}
Автоматы строятся по множеству $\B$.
\end{proof}
\begin{figure}
\centering
\includegraphics[width=0.6\textwidth]{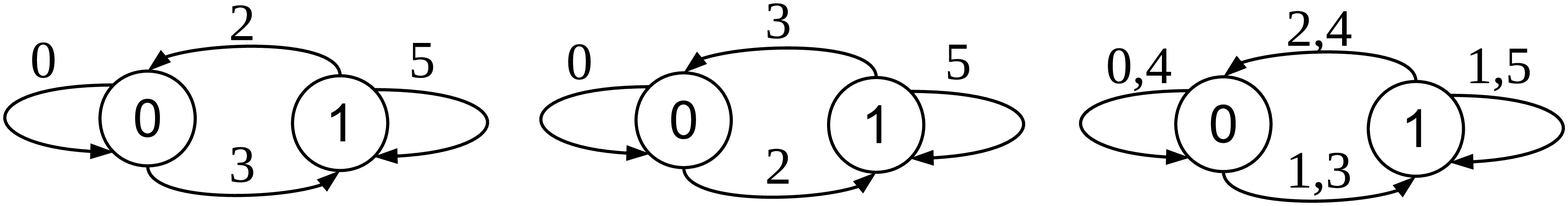}
\caption{Автоматы, порождающий все правильные слова в алфавите $\{0,2,3,5\}$ (читаемые сверху вниз -- слева, читаемые снизу вверх -- посредине) и читаемые сверху вних в алфавите $\Zed_6$ (справа). В автоматах все состояния начальные и заключительные.}
\label{I_0235.eps}
\end{figure}

\section{$n$-Мерные $0235$-продукционные пары}

Продукционные пары, крайние правые столбцы которых не содержат символов $1$ и $4$, назовём $0235$-продукционными парами.
Продукционные пары, крайние правые столбцы которых не содержат символов $1$, $4$, $5$, назовём $023$-продукционными парами.

\kurglemma\begin{lemma}[\kurglemmaa]\label{lemma:I0235table}
Список $1$-мерных $0235$-продукционных пар представляется в виде:
{\scriptsize
\[
\begin{aligned}
\I_{0235} = \begin{bmatrix}
0\\
2\\
\hline
0\\
3
\end{bmatrix}
\begin{bmatrix}
3\\
5\\
\hline
2\\
5
\end{bmatrix}
\end{aligned}
\]
}

Список $1$-мерных $023$-продукционных пар представляется в виде:
{\scriptsize
\[
\begin{aligned}
\I_{023} = \begin{bmatrix}
0\\
2\\
\hline
0\\
3
\end{bmatrix}
\begin{bmatrix}
3\\
\hline
2\\
\end{bmatrix}
\end{aligned}
\]
}
\end{lemma}

Обозначим множества $1$-мерных $0235$-продукционных и $023$-продукционных пар через $\I_{0235}$ и $\I_{023}$ соответственно.
По определению положим, что $\I_{0235}^{0\bullet}=\{0,2,3,5\}$, $\I_{0235}^{1\bullet}=\I_{0235}$.
\kurgcorollary
\begin{corollary}[\kurgcorollarya]\label{cor:VHTuples32_1c_0235}
Список всех $3$-мерных вертикальных продукционных пар с левыми столбцами в алфавите $\{0,2,3,5\}$ имеет в форме $V$-таблиц следующий вид: 
{\scriptsize
\[
\begin{aligned}
 &
\left[\begin{matrix}
0&2& \vrule & \\
0&0& \vrule & \\
0&0& \vrule & 
\end{matrix}
\begin{matrix}
0&0&0&0&1&1\\
0&0&1&1&1&2\\
0&1&1&2&2&3
\end{matrix}\right]
\left[\begin{matrix}
0&2& \vrule & \\
0&0& \vrule & \\
3&3& \vrule & 
\end{matrix}
\begin{matrix}
0&1&1&1&2&2\\
1&1&2&2&3&3\\
2&2&3&4&4&5
\end{matrix}\right]
\left[\begin{matrix}
0&2& \vrule & \\
3&3& \vrule & \\
2&2& \vrule & 
\end{matrix}
\begin{matrix}
1&1&2&2&2&3\\
2&2&3&3&4&4\\
0&1&1&2&3&3
\end{matrix}\right]
\left[\begin{matrix}
3&5& \vrule & \\
2&2& \vrule & \\
0&0& \vrule & 
\end{matrix}
\begin{matrix}
2&2&2&2&3&3\\
0&0&1&1&1&2\\
0&1&1&2&2&3
\end{matrix}\right]
\\ &
\left[\begin{matrix}
3&5& \vrule & \\
2&2& \vrule & \\
3&3& \vrule & 
\end{matrix}
\begin{matrix}
2&3&3&3&4&4\\
1&1&2&2&3&3\\
2&2&3&4&4&5
\end{matrix}\right]
\left[\begin{matrix}
0&2& \vrule & \\
3&3& \vrule & \\
5&5& \vrule & 
\end{matrix}
\begin{matrix}
2&2&3&3&3&3\\
3&4&4&4&5&5\\
2&3&3&4&4&5
\end{matrix}\right]
\left[\begin{matrix}
3&5& \vrule & \\
5&5& \vrule & \\
2&2& \vrule & 
\end{matrix}
\begin{matrix}
3&3&4&4&4&5\\
2&2&3&3&4&4\\
0&1&1&2&3&3
\end{matrix}\right]
\left[\begin{matrix}
3&5& \vrule & \\
5&5& \vrule & \\
5&5& \vrule & 
\end{matrix}
\begin{matrix}
4&4&5&5&5&5\\
3&4&4&4&5&5\\
2&3&3&4&4&5
\end{matrix}\right]
\end{aligned}
\]
}
\end{corollary}

\kurglemma\begin{lemma}[\kurglemmaa]\label{lemma:aB0235}
{\footnotesize
\[\aB\circ\I_{0235} =
\begin{aligned}
\begin{bmatrix}
0&0\\
2&2\\
\hline
0&0\\
3&0
\end{bmatrix}
\begin{bmatrix}
0&2\\
4&0\\
\hline
0&3\\
3&3
\end{bmatrix}
\begin{bmatrix}
0&3\\
2&5\\
\hline
1&5\\
4&5
\end{bmatrix}
\begin{bmatrix}
1&0\\
3&2\\
\hline
2&3\\
5&3
\end{bmatrix}
\begin{bmatrix}
1&5\\
5&3\\
\hline
2&2\\
5&2
\end{bmatrix}
\begin{bmatrix}
2&3\\
4&5\\
\hline
0&2\\
3&2
\end{bmatrix}
\begin{bmatrix}
3&0\\
5&2\\
\hline
1&0\\
4&0
\end{bmatrix}
\begin{bmatrix}
3&3\\
5&5\\
\hline
2&5\\
5&5
\end{bmatrix}
\end{aligned}
\]
}

{\footnotesize
\[\B\circ\I_{0235}=
\begin{aligned}
\begin{bmatrix}
0&0\\
2&0\\
\hline
0&0\\
3&0
\end{bmatrix}
\begin{bmatrix}
2&2\\
4&2\\
\hline
0&3\\
3&3
\end{bmatrix}
\begin{bmatrix}
0&3\\
4&3\\
\hline
0&5\\
3&5
\end{bmatrix}
\begin{bmatrix}
0&5\\
2&5\\
\hline
1&2\\
4&2
\end{bmatrix}
\begin{bmatrix}
3&0\\
5&0\\
\hline
1&3\\
4&3
\end{bmatrix}
\begin{bmatrix}
1&2\\
5&2\\
\hline
2&0\\
5&0
\end{bmatrix}
\begin{bmatrix}
1&3\\
3&3\\
\hline
2&2\\
5&2
\end{bmatrix}
\begin{bmatrix}
3&5\\
5&5\\
\hline
2&5\\
5&5
\end{bmatrix}
\end{aligned}
\]
}

{\footnotesize
\[\gB\circ\I_{0235} =
\begin{aligned}
\begin{bmatrix}
0&5\\
2&5\\
4&5\\
\hline
0&2\\
3&5
\end{bmatrix}
\begin{bmatrix}
0&3\\
2&3\\
4&3\\
\hline
0&5\\
3&2
\end{bmatrix}
\begin{bmatrix}
1&0\\
3&0\\
5&0\\
\hline
2&0\\
5&3
\end{bmatrix}
\begin{bmatrix}
1&2\\
3&2\\
5&2\\
\hline
2&3\\
5&0
\end{bmatrix}
\begin{bmatrix}
0&0\\
2&0\\
4&0\\
\hline
0&0\\
1&3\\
3&3\\
4&0
\end{bmatrix}
\begin{bmatrix}
0&2\\
2&2\\
4&2\\
\hline
0&3\\
1&0\\
3&0\\
4&3
\end{bmatrix}
\begin{bmatrix}
1&3\\
3&3\\
5&3\\
\hline
1&2\\
2&5\\
4&5\\
5&2
\end{bmatrix}
\begin{bmatrix}
1&5\\
3&5\\
5&5\\
\hline
1&5\\
2&2\\
4&2\\
5&5
\end{bmatrix}
\end{aligned}
\]

}
\end{lemma}
%
\kurglemma\begin{lemma}[\kurglemmaa]\label{lemma:VHTuples23_0235}
Множество $\B\circ\B\circ\I_{0235}$ имеет вид:
{\scriptsize
\[
\begin{aligned}
\begin{bmatrix}
0&0&0\\
2&0&0\\
\hline
0&0&0\\
3&0&0
\end{bmatrix}
\begin{bmatrix}
0&0&5\\
4&0&5\\
\hline
0&1&2\\
3&1&2
\end{bmatrix}
\begin{bmatrix}
0&1&2\\
2&1&2\\
\hline
0&2&0\\
3&2&0
\end{bmatrix}
\begin{bmatrix}
0&2&0\\
4&2&0\\
\hline
0&3&0\\
3&3&0
\end{bmatrix}
\begin{bmatrix}
0&2&2\\
2&2&2\\
\hline
0&3&3\\
3&3&3
\end{bmatrix}
\begin{bmatrix}
0&3&0\\
4&3&0\\
\hline
0&4&3\\
3&4&3
\end{bmatrix}
\begin{bmatrix}
0&3&3\\
2&3&3\\
\hline
0&5&2\\
3&5&2
\end{bmatrix}
\begin{bmatrix}
0&4&2\\
4&4&2\\
\hline
1&0&3\\
4&0&3
\end{bmatrix}
\\ 
\begin{bmatrix}
0&4&3\\
2&4&3\\
\hline
1&0&5\\
4&0&5
\end{bmatrix}
\begin{bmatrix}
0&5&2\\
4&5&2\\
\hline
1&2&0\\
4&2&0
\end{bmatrix}
\begin{bmatrix}
0&5&5\\
2&5&5\\
\hline
1&2&5\\
4&2&5
\end{bmatrix}
\begin{bmatrix}
1&0&3\\
5&0&3\\
\hline
1&3&5\\
4&3&5
\end{bmatrix}
\begin{bmatrix}
1&0&5\\
3&0&5\\
\hline
1&4&2\\
4&4&2
\end{bmatrix}
\begin{bmatrix}
1&1&3\\
5&1&3\\
\hline
1&5&2\\
4&5&2
\end{bmatrix}
\begin{bmatrix}
1&2&0\\
3&2&0\\
\hline
2&0&0\\
5&0&0
\end{bmatrix}
\begin{bmatrix}
1&2&5\\
5&2&5\\
\hline
2&1&2\\
5&1&2
\end{bmatrix}
\\ 
\begin{bmatrix}
1&3&0\\
3&3&0\\
\hline
2&1&3\\
5&1&3
\end{bmatrix}
\begin{bmatrix}
1&3&5\\
5&3&5\\
\hline
2&2&5\\
5&2&5
\end{bmatrix}
\begin{bmatrix}
1&4&2\\
3&4&2\\
\hline
2&3&3\\
5&3&3
\end{bmatrix}
\begin{bmatrix}
1&5&0\\
5&5&0\\
\hline
2&4&3\\
5&4&3
\end{bmatrix}
\begin{bmatrix}
1&5&2\\
3&5&2\\
\hline
2&5&0\\
5&5&0
\end{bmatrix}
\begin{bmatrix}
2&0&3\\
4&0&3\\
\hline
0&0&5\\
3&0&5
\end{bmatrix}
\begin{bmatrix}
2&1&3\\
4&1&3\\
\hline
0&2&2\\
3&2&2
\end{bmatrix}
\begin{bmatrix}
2&2&5\\
4&2&5\\
\hline
0&4&2\\
3&4&2
\end{bmatrix}
\\ 
\begin{bmatrix}
2&3&5\\
4&3&5\\
\hline
0&5&5\\
3&5&5
\end{bmatrix}
\begin{bmatrix}
2&5&0\\
4&5&0\\
\hline
1&1&3\\
4&1&3
\end{bmatrix}
\begin{bmatrix}
3&0&0\\
5&0&0\\
\hline
1&3&0\\
4&3&0
\end{bmatrix}
\begin{bmatrix}
3&1&2\\
5&1&2\\
\hline
1&5&0\\
4&5&0
\end{bmatrix}
\begin{bmatrix}
3&2&2\\
5&2&2\\
\hline
2&0&3\\
5&0&3
\end{bmatrix}
\begin{bmatrix}
3&3&3\\
5&3&3\\
\hline
2&2&2\\
5&2&2
\end{bmatrix}
\begin{bmatrix}
3&4&3\\
5&4&3\\
\hline
2&3&5\\
5&3&5
\end{bmatrix}
\begin{bmatrix}
3&5&5\\
5&5&5\\
\hline
2&5&5\\
5&5&5
\end{bmatrix}
\end{aligned}
\]
}
\end{lemma}

%
\kurgtheorem\begin{theorem}[\kurgtheorema]\label{theorem:HtablepropertyB}
$H$-таблицы, представляющие множество  $\B^{(n-1)\circ}\circ\I_{0235}$, имеют уникальные верхние и нижние компоненты, в которых ровно по две строки. В таком представлении $2\cdot 6^{n-1}$ таблиц. В каждой таблице строки $x_nx_{n-1}\ldots x_1$ и $x'_{n}x'_{n-1}\ldots x'_1$  верхней компоненты, как и строки $y_ny_{n-1}\ldots y_1$ и $y'_{n}y'_{n-1}\ldots y'_1$  нижней компоненты отличаются только в крайнем левом столбце. При этом, если  $x_i,x'_i\in\{1,3,5\}$, то  $y_i, y'_i\in \{1,2,4,5\}$,  а если  $x_i,x'_i\in\{0, 2, 4\}$, то  $y_i, y'_i\in \{0, 1, 3, 4\}$.
\end{theorem}
\begin{proof}
Доказательство в форме математической индукции легко следует из леммы~\ref{lemma:aB0235} для  $\B\circ\I_{0235}$ и леммы~\ref{lemma:B} для $\B$.
\end{proof}
\kurgcorollary\begin{corollary}[\kurgcorollarya]
В множестве  $\B^{(n-1)\circ}\circ\I_{0235}$ ровно $4\cdot 6^{n-1}$ различных строк.
\end{corollary}
\kurglemma\begin{lemma}[\kurglemmaa]\label{lemma:VHTuplesDiag23_0235}
Список элементов множества $\aB\circ\aB\circ\I_{0235}$ имеет вид:
{\scriptsize
\[
\begin{aligned}
\begin{bmatrix}
0&0&0\\
2&2&2\\
\hline
0&0&0\\
3&0&0
\end{bmatrix}
\begin{bmatrix}
0&0&2\\
4&4&0\\
\hline
0&0&3\\
3&0&3
\end{bmatrix}
\begin{bmatrix}
0&0&3\\
2&2&5\\
\hline
0&1&5\\
3&1&5
\end{bmatrix}
\begin{bmatrix}
0&1&0\\
2&3&2\\
\hline
0&2&3\\
3&2&3
\end{bmatrix}
\begin{bmatrix}
0&1&5\\
4&5&3\\
\hline
0&2&2\\
3&2&2
\end{bmatrix}
\begin{bmatrix}
0&2&2\\
4&0&0\\
\hline
0&3&0\\
3&3&0
\end{bmatrix}
\begin{bmatrix}
0&2&3\\
2&4&5\\
\hline
0&3&2\\
3&3&2
\end{bmatrix}
\begin{bmatrix}
0&2&5\\
4&0&3\\
\hline
1&4&5\\
4&4&5
\end{bmatrix}
\\ 
\begin{bmatrix}
0&3&0\\
2&5&2\\
\hline
1&4&0\\
4&4&0
\end{bmatrix}
\begin{bmatrix}
0&3&2\\
4&1&0\\
\hline
1&5&3\\
4&5&3
\end{bmatrix}
\begin{bmatrix}
0&3&3\\
2&5&5\\
\hline
1&5&5\\
4&5&5
\end{bmatrix}
\begin{bmatrix}
1&0&0\\
3&2&2\\
\hline
2&3&0\\
5&3&0
\end{bmatrix}
\begin{bmatrix}
1&0&2\\
5&4&0\\
\hline
2&3&3\\
5&3&3
\end{bmatrix}
\begin{bmatrix}
1&0&3\\
3&2&5\\
\hline
2&4&5\\
5&4&5
\end{bmatrix}
\begin{bmatrix}
1&1&0\\
3&3&2\\
\hline
2&5&3\\
5&5&3
\end{bmatrix}
\begin{bmatrix}
1&1&5\\
5&5&3\\
\hline
2&5&2\\
5&5&2
\end{bmatrix}
\\ 
\begin{bmatrix}
1&4&0\\
3&0&2\\
\hline
1&0&3\\
4&0&3
\end{bmatrix}
\begin{bmatrix}
1&4&5\\
5&2&3\\
\hline
1&0&2\\
4&0&2
\end{bmatrix}
\begin{bmatrix}
1&5&2\\
5&3&0\\
\hline
1&1&0\\
4&1&0
\end{bmatrix}
\begin{bmatrix}
1&5&3\\
3&1&5\\
\hline
2&2&2\\
5&2&2
\end{bmatrix}
\begin{bmatrix}
1&5&5\\
5&3&3\\
\hline
2&2&5\\
5&2&5
\end{bmatrix}
\begin{bmatrix}
2&2&3\\
4&4&5\\
\hline
0&0&2\\
3&0&2
\end{bmatrix}
\begin{bmatrix}
2&3&0\\
4&5&2\\
\hline
0&1&0\\
3&1&0
\end{bmatrix}
\begin{bmatrix}
2&3&3\\
4&5&5\\
\hline
0&2&5\\
3&2&5
\end{bmatrix}
\\ 
\begin{bmatrix}
2&4&0\\
4&0&2\\
\hline
0&3&3\\
3&3&3
\end{bmatrix}
\begin{bmatrix}
2&5&3\\
4&1&5\\
\hline
1&5&2\\
4&5&2
\end{bmatrix}
\begin{bmatrix}
3&0&0\\
5&2&2\\
\hline
1&0&0\\
4&0&0
\end{bmatrix}
\begin{bmatrix}
3&0&3\\
5&2&5\\
\hline
1&1&5\\
4&1&5
\end{bmatrix}
\begin{bmatrix}
3&1&0\\
5&3&2\\
\hline
2&2&3\\
5&2&3
\end{bmatrix}
\begin{bmatrix}
3&2&3\\
5&4&5\\
\hline
2&3&2\\
5&3&2
\end{bmatrix}
\begin{bmatrix}
3&3&0\\
5&5&2\\
\hline
2&4&0\\
5&4&0
\end{bmatrix}
\begin{bmatrix}
3&3&3\\
5&5&5\\
\hline
2&5&5\\
5&5&5
\end{bmatrix}
\end{aligned}
\]
}
\end{lemma}

Множество $\aB^{(n-1)\circ}\circ\I_{0235}$ получается из представления $\aB^{(n-1)}$ в форме $H$-таблиц удалением всех $H$-$14$-таблиц и удалением в каждой $H$-$0235$-таблице из верхней компоненты единственной строки, заканчивающейся на $1$ или $4$. Отсюда следует

\kurgtheorem\begin{theorem}[\kurgtheorema]\label{theorem:Htableproperty}
$H$-таблицы, представляющие множество $\aB^{(n-1)\circ}\circ\I_{0235}$, имеют уникальные верхние и нижние компоненты, в которых ровно по две строки. Всего в представлении $2\cdot 4^{n-1}$ таблиц. В каждой таблице строки $x_nx_{n-1}\ldots x_1$ и $x'_{n}x'_{n-1}\ldots x'_1$  верхней компоненты отличаются в каждом столбце, при этом $x_i,x'_i\in\{1,3,5\}$ или $x_i,x'_i\in\{0,2,4\}$, $1\le i \le n$. В каждой таблице строки $y_ny_{n-1}\ldots y_1$ и $y'_{n}y'_{n-1}\ldots y'_1$  нижней компоненты отличаются только в крайнем левом столбце, при этом $y_n, y'_n\in \{0,3\}$ или $y_n, y'_n\in \{2,5\}$ или $y_n, y'_n\in \{1,4\}$. Если  $x_i,x'_i\in\{1,3,5\}$, то  $y_i, y'_i\in \{1,2,4,5\}$,  а если  $x_i,x'_i\in\{0, 2, 4\}$, то  $y_i, y'_i\in \{0, 1, 3, 4\}$.
\end{theorem}
\kurgcorollary\begin{corollary}[\kurgcorollarya]
В множестве $\aB^{(n-1)\circ}\circ\I_{0235}$ ровно $4^n$ различных строк.
\end{corollary}
\kurgcorollary\begin{corollary}[\kurgcorollarya]\label{corollar:vyvod_vniz_BaB}
Верны следующие утверждения.
\begin{enumerate}
\item\label{corollar:vyvod_vniz_BaB1}
Если
\[
\begin{pmatrix}
c_{n}\ldots c_1c_0\\
a_{n}\ldots a_1a_0
\end{pmatrix}\in \aB^{n\circ}\circ\I_{0235},\quad
\begin{pmatrix}
c_{n}\ldots c_1c_0\\
b_{n}\ldots b_1b_0
\end{pmatrix}\in \aB^{n\circ}\circ\I_{0235},
\]
 то $a_{n-1}\ldots a_1a_0
=b_{n-1}\ldots b_1b_0$.
\item\label{corollar:vyvod_vniz_BaB2}
Если
\[
\begin{pmatrix}
c_{n}\ldots c_1c_0\\
a_{n}\ldots a_1a_0
\end{pmatrix}\in {\B}^{n\circ}\circ\I_{0235},\quad
\begin{pmatrix}
c_{n}\ldots c_1c_0\\
b_{n}\ldots b_1b_0
\end{pmatrix}\in {\B}^{n\circ}\circ\I_{0235},
\]
то $a_{n-1}\ldots a_1a_0
=b_{n-1}\ldots b_1b_0$.
\end{enumerate}
\end{corollary}
\begin{figure}
\centering
\includegraphics[width=0.4\textwidth]{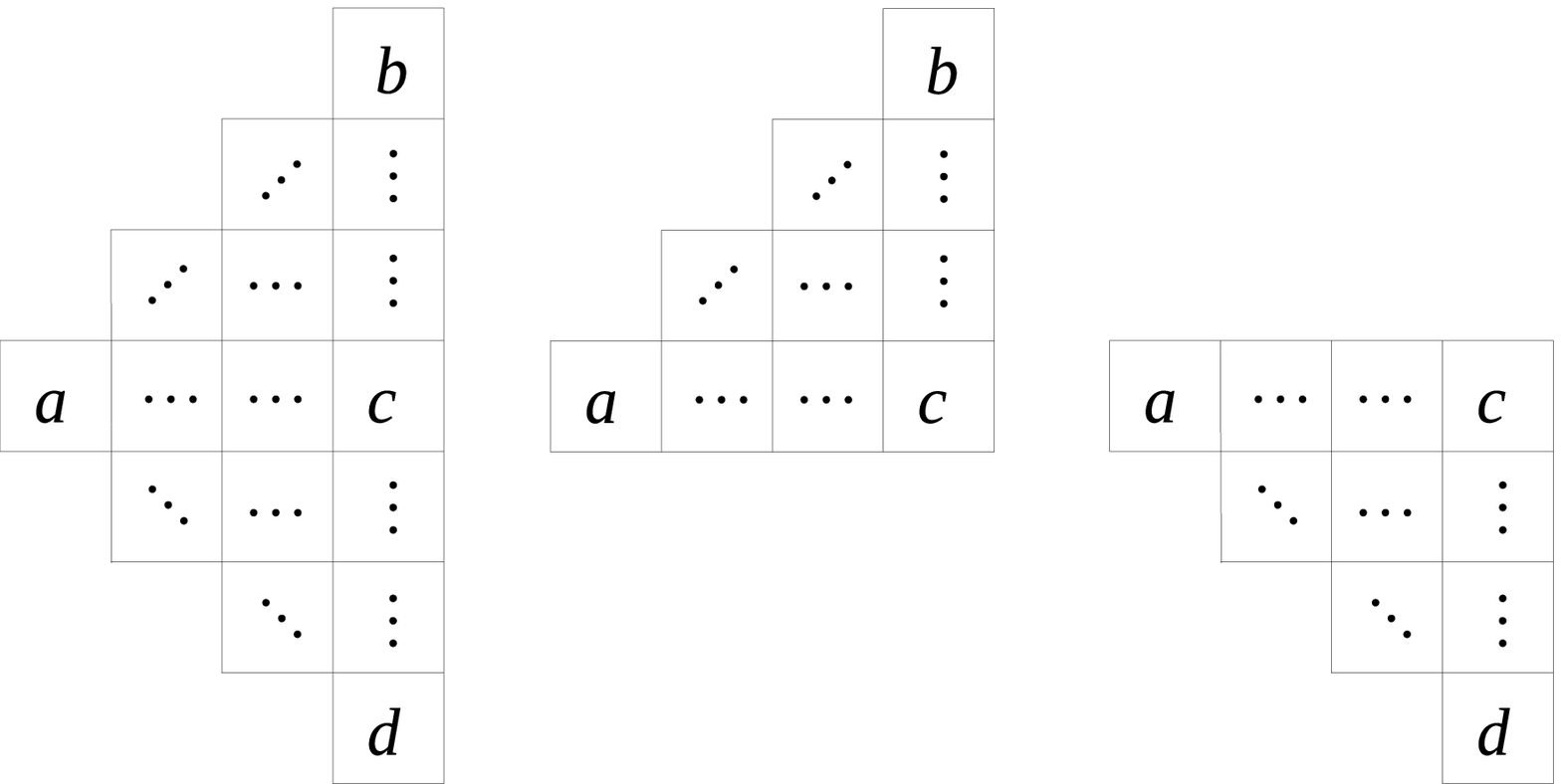}
\caption{Иллюстрация вида фрагментов слов к теореме~\ref{theorem:FigureKurgPotap2018_triangle}.} 
\label{FigureKurgPotap2018_triangle.eps}
\end{figure}
\kurgtheorem\begin{theorem}[\kurgtheorema]\label{theorem:FigureKurgPotap2018_triangle}
Пусть  $T$ -- множество $2$-мерных треугольных фрагментов слов из $\Univ$ как показано на рисунке~\ref{FigureKurgPotap2018_triangle.eps} (слева), в которых вертикальное слово $b\ldots c \ldots d$ не содержит символов $1$ и $4$. Множество $T$ определяет взаимно-однозначное соответствие между множеством вертикальных слов $b\ldots c \ldots d$ и множеством диагональных слов $a\ldots b$ в $T$.
\end{theorem}
\begin{proof}
По следствию~\ref{corollar:vyvod_vniz_BaB}, пункт~\ref{corollar:vyvod_vniz_BaB1}, слово $a\ldots b$ полностью и однозначно определяет фрагмент, показанный на рисунке~\ref{FigureKurgPotap2018_triangle.eps} в центре.
По следствию~\ref{corollar:vyvod_vniz_BaB}, пункт~\ref{corollar:vyvod_vniz_BaB2},  слово $a\ldots c$ полностью и однозначно определяет фрагмент, показанный на рисунке~\ref{FigureKurgPotap2018_triangle.eps} справа.\qed
То, что слово $b\ldots d$ однозначно определяет всё треугольное слово, следует из леммы~\ref{lemma:leftableitung}, пункт~\ref{lemma:leftableitung1}.
\end{proof}
\kurgcorollary
\begin{corollary}[\kurgcorollarya]
Множество квадратных слов $\cup_{n=1}^{\infty}{(\aB^{n\circ}\circ\I_{0235})^{n\bullet}}$ устанавливает взаимно-однозначное соответствие между его нулевыми (крайними верхними) строками и нулевыми (крайними правыми) столбцами.
\end{corollary}
\kurgcorollary
\begin{corollary}[\kurgcorollarya]\label{cor:1to1}
Множество слов $(\aB^{+\omega\circ}\circ\I_{0235})^{+\omega\bullet}$ устанавливает взаимно-однозначное соответствие между его нулевыми строками и столбцами.
\end{corollary}

Обозначим здесь множество всех строк слов из $(\aB^{+\omega\circ}\circ\I_{0235})^{+\omega\bullet}$ через $A$. Обозначаем дальше биекцию из следствия~\ref{cor:1to1} через $\Bi:A\rightarrow \I_{0235}^{+\omega\bullet}$.

\kurgtheorem\begin{theorem}[\kurgtheorema]
Для любого $n\in\Nat$ язык всех правильных слов в алфавите $\{0,2,3,5\}$ включается в язык крайних левых столбцов слов из $\bigcup_{m=1}^{\infty}(\aB^{n\circ})^{m\bullet}$.
\end{theorem}
\begin{proof}
Утверждение следует из теоремы~\ref{theorem:basictable} о структуре $H$-таблиц для $\aB^{n\circ}$. Удалим из $H$-таблиц все строки с цифрами $1$ и $4$ в крайнем левом столбце. В верхних компонентах останутся ровно $2$ строки с цифрами $0$ и $2$ или $3$ и $5$ в крайних левых столбцах. Некоторые нижние компоненты окажутся пустыми. Удалим соответствующие таблицы как тупиковые. Полученная совокупность таблиц представляет некоторую конечную совокупность продукционных пар $W$, в которых множество строк верхних компонент строго включается в множество строк нижних. Это значит, что для любого $w\in W$ $w\bullet W\neq\emptyset$. 
Так как в верхних компонентах в левом столбце всегда ровно два символа $0$ и $2$ или $3$ и $5$, то построив по $W$
минимальный автомат, допускающий язык читаемых снизу вверх левых столбцов, мы получим автомат на среднем рисунке~\ref{I_0235.eps}.
\qed
\end{proof}
\kurgcorollary\begin{corollary}[\kurgcorollarya]\label{corollary:u0235w_in_Univ}
Для любого слова $u\in (\B^{+\omega\circ}\circ\I_{0235})^{\pm\omega\bullet}$  существует слово $w\in (\I_{0235}\circ\B^{-\omega\circ})^{\pm\omega\bullet}$ такое, что $u\circ w\in \Univ$.
\end{corollary}

\kurgtheorem\begin{theorem}[\kurgtheorema]\label{theorem:basic0235p}\label{theorem:tBbasic0235p}
Пусть $U = U_1 = (\B^{+\omega\circ}\circ\I_{0235}\circ\B^{-\omega\circ})^{\pm\omega\bullet}$ или  $U = U_2 =(\B^{+\omega\circ}\circ\I_{0235})^{\pm\omega\bullet}$ или $U = U_3 =(\aB^{+\omega\circ}\circ\I_{0235})^{+\omega\bullet}$. Пусть $i\ge 0$, $W_i=\{u[:,i]|u\in  U\}$. Тогда ограничение отображения $L$ на множество $W_i$, то есть $L:W_i\rightarrow W_{i+1}$, является биекцией.
\end{theorem}
\begin{proof}
Теорему достаточно доказать случай $U = U_1$. Из него по следствию~\ref{corollary:u0235w_in_Univ} получим $U_2$ и $U_3$.
При этом достаточно доказать инъективность $L$.

Пусть $u, u'\in U_1$ и для некоторого $i\ge 0$ $i$-ые столбцы слов $u$ и $u'$ не равны, $u[:,i] \neq u'[:,i]$.  Предположим, что теорема не верна и $L(u[:,i])= L(u'[:,i])$, то есть $u[:,i+n]=u'[:,i+n]$, $n>0$. Но тогда $u[j,i:]=u'[j,i:]$ или, другими словами, рассматривая строки двумерных слов как числовые объекты в системе счисления по основанию $6$ с евклидовым расстоянием, $\left|u[j,:]-u'[j,:]\right|<6^{i+1}$ для любого $j\in \Zed$.

Пусть для некоторого $j\in\Zed$ $u[j,0]\neq u'[j,0]$. Не нарушая общности рассуждений положим $j=0$, а также $u[0,0] > u'[0,0]$. Если показать, что $u[0,:]-u'[0,:] \neq 0$, 
то при достаточно большом $n$ мы получим $\left|u[n,:]-u'[n,:]\right|=\left|\left(\frac{3}{2}\right)^n(u[0,:]-u'[0,:])\right|>6^{i+1}$, а это противоречит тому, что для всех $j\in \Zed$ должно выполняться равенство  $u[j,i:]=u'[j,i:]$. Отсюда будет следовать $L(u[:,i])\neq L(u'[:,i])$.

Итак, покажем, что $u[0,:]-u'[0,:] \neq 0$. Заметим, что из $u[0,:]-u'[0,:] = 0$ следует $u[j,:]-u'[j,:] = 0$ для любого $j\in\Zed$.
Из равенства чисел с разными цифрами в нулевом разряде необходимо следует одна из следующих ситуаций. Либо $u[0,0]=5$ и все символы справа от $u[0,0]$ равны $5$, то есть $u[0,:1]=55\ldots 5\ldots$, и $u'[0,:1] = 00\ldots 0\ldots$. Либо $u[0,:1]=30\ldots 0\ldots$ и  $u'[0,:1] = 25\ldots 5\ldots$.

Из следствия~\ref{cor:VHTuples32_1c_0235}, учитывая алфавит $\{0,2,3,5\}$ нулевого столбца, следует, что ситуация $u[0,:1]=30\ldots 0\ldots$ и $u'[0,:1] = 25\ldots 5\ldots$ невозможна.

Рассмотрим ситуацию $u[0,0]=5$, $u'[0,0] = 0$. 
Найдем наименьший положительный индекс $j\in\Zed$ такой, что  $u[j,0]\neq 5$ или $u'[j,0] \neq 0$. Если же такого $j$ нет, то все столбцы слева от $u[:,0]$ начиная с некоторого момента, как можно увидеть из леммы~\ref{lemma:U32}, состоят из $5$, а все столбцы слева от $u'[:,0]$, начиная с некоторого момента, состоят из $0$, то есть $L(u[:,i])\neq L(u'[:,i])$, и мы сразу приходим к противоречию с предположением $L(u[:,i]) = L(u'[:,i])$.

Предположим, что $u[j,0]\neq 5$ и $j>0$. Тогда $u[j,0] = 2$. Далее возможны только две ситуации: $u'[j,0] = 0$ или $u'[j,0] = 3$. В обоих случаях $u[j,0] - u'[j,0] \neq 0$.

Предположим, что $u'[j,0]\neq 0$ и $j>0$. Тогда $u[j,0] = 3$. Далее возможны только две ситуации: $u[j,0] = 5$ или $u[j,0] = 2$. В обоих случаях $u[j,0] - u'[j,0] \neq 0$.

Предположим, что $u[j,0]\neq 5$ и $j<0$. Тогда $u[j,0] = 3$. Далее возможны только две ситуации: $u'[j,0] = 0$ или $u'[j,0] = 2$. В обоих случаях $u[j,0] - u'[j,0] \neq 0$.

Предположим, что $u'[j,0]\neq 0$ и $j<0$. Тогда $u[j,0] = 2$. Далее возможны только две ситуации: $u[j,0] = 5$ или $u[j,0] = 3$. В обоих случаях $u[j,0] - u'[j,0] \neq 0$.
\qed\end{proof}

\section{Графы $n$-мерных $0235$-продукционных пар}
Зафиксируем натуральное $n>1$. Обозначим $\Sn_n=\aB^{(n-1)\circ}\circ\I_{0235}$. Множество всех различных строк в словах из $\Sn_n$ обозначим через $\A_n$. Дальше мы не будем указывать нижний индекс $n$ в обозначении $\Sn_n$ и $\A_n$, пока в этом не возникнет необходимость. Рассматриваем $\Sn$ как бинарное отношение на $\A$, то есть $w\Sn w'$ означает то же самое, что и $\begin{pmatrix}w\\w'\end{pmatrix}\in \Sn$.

Определим граф $G=(\A, \Sn)$ в котором $\A$~-- множество вершин, $\Sn$~-- множество дуг. В графе $G$ каждая вершина имеет ровно две входящие и две выходящие дуги. Вершину графа, состоящую из одних $0$, назовем нулевой вершиной. Вершины назовём двойственными, если они соответствуют двойственным словам. Через $G^{-1}$ обозначим граф, полученный из $G$ заменой направления всех дуг на обратные.
Поскольку $\tB^{(n-1)\circ}=(\tB^{(n-1)\circ})^{dual}$, то верна

\kurglemma\begin{lemma}[\kurglemmaa]
Взаимно-однозначное отображение $\phi:\A\rightarrow \A$ вершин графа $G$, связывающее двойственные вершины, является изоморфизмом графов $G \simeq \phi(G)$.
\end{lemma}

Очевидным образом, исходя из $H$-табличного представления продукционных пар и теоремы~\ref{theorem:Htableproperty}, верна следующая лемма.

\kurglemma\begin{lemma}[\kurglemmaa]\label{lemma:grafsstt}
Если для вершин $s$, $s'$, $t$ графа $G$ верно, что $(s,t)$ и $(s',t)$ являются дугами $G$, то вторые дуги, исходящие из $s$ и $s'$, также ведут в одну вершину, то есть существует $t'$, что $(s,t')$ и $(s',t')$ являются дугами.
\end{lemma}

Определим преобразование $\rho$ графа $G$, которое назовем свёрткой. 

\begin{definition}\label{definition:reduce}
Разобьём вершины графа $G$ на пары $\{s,s'\}$, исходящие дуги которых имеют общие концы. По этим парам построим новый граф такой, что из пары $\{s_0,s'_0\}$ дуга ведёт в $\{s_1,s'_1\}$ тогда и только тогда, когда в графе $G$ существует либо дуга $(s_0,s_1)$, либо $(s_0,s'_1)$ или, что то же самое, существует либо дуга $(s'_0,s_1)$, либо $(s'_0,s'_1)$. Другим словами, свёртка $\rho(G)$ получается из $G$ склейкой вершин $s$ и $s'$, у которых совпадают концы исходящих дуг.
\end{definition}

Вершины графа $\rho(G)$ представляем парами строк из $\A$. 
Пример графа $\rho(G)$ для $n=2$ показан на рисунке~\ref{GrafDiag22.pdf}.
\begin{figure}
\centering
\includegraphics[width=0.4\textwidth]{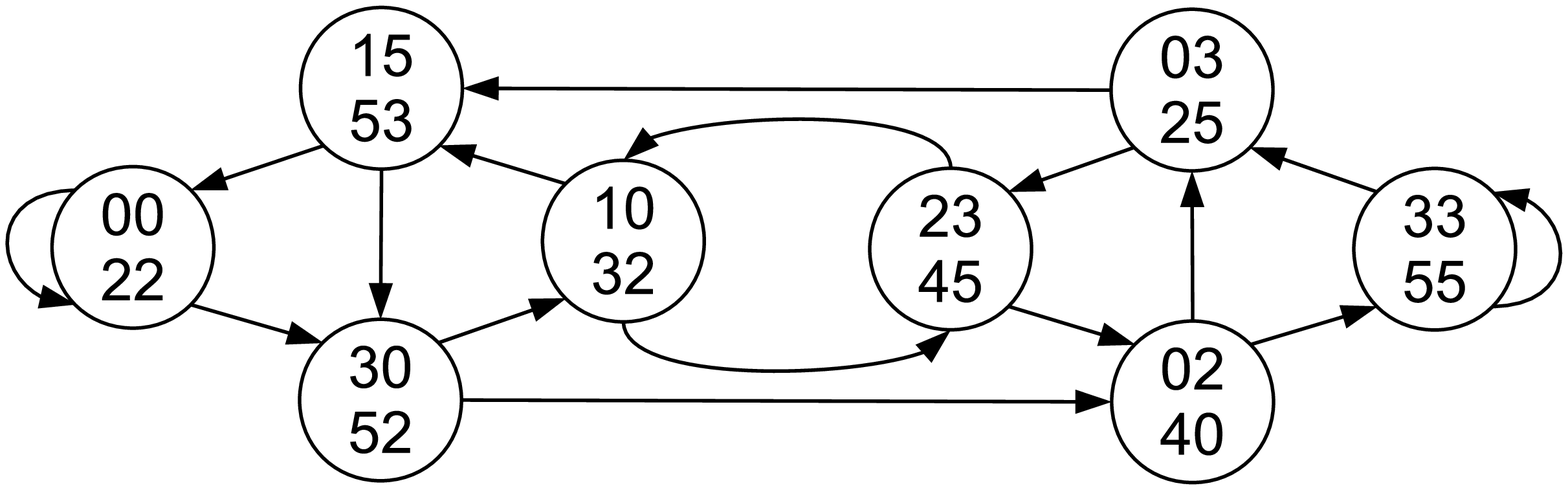}
\caption{Граф $\rho(G)$ для $n=2$.} 
\label{GrafDiag22.pdf}
\end{figure}

Лемма~\ref{lemma:grafsstt} верна и для графа $\rho(G)$, то есть к графу $\rho(G)$ применима операция свёртки $\rho^2(G)=\rho(\rho(G))$. Более того $\rho^2(G_{n+1})\simeq G_{n}$ и $G_{n}\simeq G^{-1}_{n}$, $n\ge 1$. Прежде чем это доказать, рассмотрим граф $G$ с другой точки зрения.

Исходя из леммы~\ref{lemma:I0235table} построим автомат, допускающий язык $\bigcup_{n=0}^{\infty}\I_{0235}^{n\bullet}\cup\{\lambda\}$ всех правильных столбцов в алфавите $\{0,2,3,5\}$. 
Автомат представлен  на рисунке~\ref{I_0235.eps} слева.
%
%
По этому языку и $n$ построим следующий граф $\Gamma_n$. Множеством его вершин является $\I_{0235}^{(2n-2)\bullet}$. 
Пусть $w=x_{2n-1}\ldots x_2x_1$, $w'=x'_{2n-1}\ldots x'_2x'_1\in\I_{0235}^{(2n-2)\bullet}$. Упорядоченная пара $(w,w')$ является дугой графа $\Gamma_n$ тогда и только тогда, когда $x_{2n-2}\ldots x_2x_1=x'_{2n-1}\ldots x'_2$. Из определения следует, что $\Gamma_n$ является графом де Брёйна.
В силу теоремы~\ref{theorem:FigureKurgPotap2018_triangle} верно
\kurgcorollary\begin{corollary}[\kurgcorollarya]\label{corollary:ggammaisomorph}
Графы $G_n$ и $\Gamma_n$ изоморфны, $n\geq 0$.
\end{corollary}

Для графа $\Gamma_n$ очевиден следующий переход к $\Gamma_{n+1}$ за два шага. На первом шаге  строится промежуточный граф $\Gamma_{n+\frac{1}{2}}$. Его вершинами являются дуги графа $\Gamma_n$. Пусть $w$, $w'$ вершины графа $\Gamma_{n+\frac{1}{2}}$. Пара $(w,w')$ является дугой графа $\Gamma_{n+\frac{1}{2}}$ тогда и только тогда, когда в графе $\Gamma_n$ конец дуги $w$ является началом дуги $w'$. 
%
\begin{figure}
\centering
\includegraphics[width=0.8\textwidth]{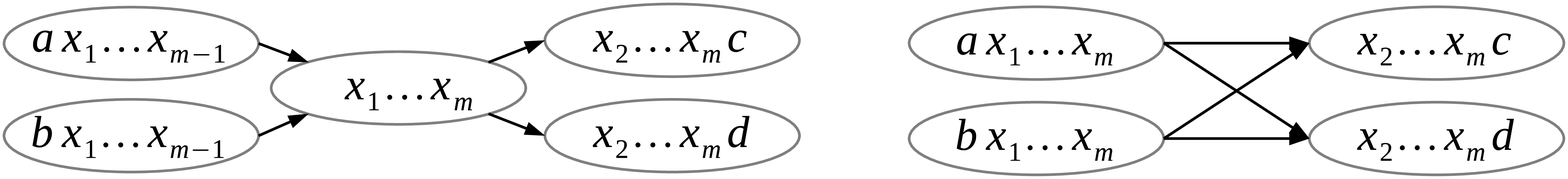}
\caption{Фрагмент перехода от графа $\Gamma_n$ (слева) к графу $\Gamma_{n+\frac{1}{2}}$ (справа), $m=2n-1$.}
\label{gammaG.eps}
\end{figure}
На рисунке~\ref{gammaG.eps} показан фрагмент перехода от графа $\Gamma_n$ к  $\Gamma_{n+\frac{1}{2}}$. Видно, что вершины $ax_1\ldots x_n$ и $bx_1\ldots x_n$ удовлетворяют свойствам, сформулированным в лемме~\ref{lemma:grafsstt}. Склейка вершин 
вида $ax_1\ldots x_n$ и $bx_1\ldots x_n$ в вершину $x_1\ldots x_n$ является операцией свёртки. Таким образом верна теорема.

\kurgtheorem\begin{theorem}[\kurgtheorema]\label{theorem:reduce}
Верно, что $\rho^2(G_n)\simeq G_{n-1}$, $n>1$.
\end{theorem}

В силу того, что графы $G_1$ и $G^{-1}_1$ изоморфны и построение графа $G_i$ из $G_{i-1}$ зависит только от структуры графа $G_{i-1}$ и не зависит от имён вершин, верно следующее утверждение:

\kurgtheorem\begin{theorem}[\kurgtheorema]
Графы $G_n$ и $G^{-1}_n$ изоморфны.
\end{theorem}

\section{Автоматы $n$-мерных $0235$-продукционных пар}
По графу $G$ определим автомат $A=(\A, \Zed_6, \delta_A)$, где $\A$~-- множество состояний, $\Zed_6$~-- входной алфавит, $\delta_{A}:\A\times \Zed_6\rightarrow \A$~--  функция переходов такая, что для любых $s,s'\in \A$ и $x\in\Zed_6$ $\delta_{A}(s,x)=s'$ тогда и только тогда, когда $(s,s')$ является дугой графа $G$ и $s'=xx_2\ldots x_n$.
Язык без пустого слова, допускаемый состоянием $s\in \A$ автомата $A$, когда все состояния заключительные, обозначим через $\lambda_{A}(s)$. Через $\rho(A)$ обозначим автомат, полученный из $A$ свёрткой. Заметим, что склеиваемые состояния при свёртке имеют в автомате $A$ одинаковые значения функции $\delta_A$.

\kurgtheorem
\begin{theorem}[\kurgtheorema]\label{theorem:basicA_n}
$s=s'$, $s,s'\in \rho(A)$, если, и только если, $\lambda_{\rho(A)}(s) = \lambda_{\rho(A)}(s')$.
\end{theorem}
\begin{proof}
Cледует из теоремы~\ref{theorem:tBbasic0235p}.
\qed
\end{proof}

Построим по автомату $A$ эквивалентный ему по допускаемому языку минимальный инициальный (с одним начальным состоянием) детерминированный автомат $D$. Ядром $K$ автомата $D$ назовем его подавтомат, образованный состояниями, принадлежащими сильно-связным компонентам графа автомата $D$, то есть в ядро входят состояния, через которые проходят циклические пути.

\kurglemma\begin{lemma}[\kurglemmaa]
Ядро $K$ и автомат $A$ изоморфны.
\end{lemma}
\begin{proof}
Оба автомата минимальны и допускают одно и то же множество языков.
\qed
\end{proof}

Обозначим через $\eta_n$ минимальное положительное целое такое, что для всех слов $w\in\Zed_6^*$ длины $\eta_n$ в автомате $D$  любой путь длины больше либо равной $\eta_n$ ведёт из начальной вершины в ядро. Обозначим через $\theta_n$ минимальное положительное целое такое, слово $0^{\theta_n}$ в автомате $D$ ведёт из начальной вершины в ядро.

\kurglemma
\begin{lemma}[\kurglemmaa]
Пусть $w_1$, $w_2\in (\tB^{(n-1)\circ}\circ\I_{0235})^{m\bullet}$, $m\in\Nat$, такие, что $w_1[:,n-1] = w_2[:,n-1]$ и слово $w_1[:,n-1]$ ведет в автомате $D$ из начального состояния в ядро. Тогда $w_1=w_2$.
\end{lemma}
\begin{proof}
Предположим $w_1[:,n-2]\neq w_2[:,n-2]$. Обозначим $\tilde{a}' = w_1[:,n-1]$. 
Пусть $s_0$~-- начальное состояние автомата $D$, а $t$~-- состояние ядра такое, что $\delta_{D}(s_0,\tilde{a}) = t$.
Пусть $\tilde{a}''$ произвольное бесконечное слово, порождаемое состоянием $t$. Тогда конкатенация $\tilde{a}$ слов $\tilde{a}'$ и $\tilde{a}''$ принадлежит $(\tB^{(n-1)\circ}\circ\I_{0235})^{+\omega\bullet}[:,n-1]$ и существуют два слова $w'_1$, $w''_2\in (\tB^{(n-1)\circ}\circ\I_{0235})^{+\omega\bullet}$ с одинаковыми крайними левыми столбцами $\tilde{a}$, рядом с которыми в $w'_1$ и $w''_2$ справа стоят разные столбцы, начинающиеся на $w_1[:,n-2]$ и $w_2[:,n-2]$ соответственно, что противоречит теореме~\ref{theorem:basic0235p}.
\end{proof}

\kurgcorollary
\begin{corollary}[\kurgcorollarya]
Если $w_1$, $w_2\in (\tB^{(n-1)\circ}\circ\I_{0235})^{(\eta_n-1)\bullet}$ и $w_1[:,n-1] = w_2[:,n-1]$, то $w_1=w_2$.
\end{corollary}
\kurgcorollary
\begin{corollary}[\kurgcorollarya]\label{cor:theta}
Если $w\in (\tB^{(n-1)\circ}\circ\I_{0235})^{(\theta_n-1)\bullet}$ и столбец $w[:,n-1]$ состоит из нулей, то всё слово $w$ состоит из нулей.
\end{corollary}

\section{Достаточные условия пустоты множества $Z$-чисел}

Пусть $\xi\in\Real$ произвольное $Z$-число. Рассмотрим его как элемент языка $\Zed_6^{\pm\omega}$, то есть как бесконечную в обе стороны (влево и вправо) последовательность цифр с разделяющим на целую и дробную части знаком. Отсутствующие цифры заполним нулями. Бесконечная в обе стороны (вверх и вниз) последовательность чисел $\{\xi(3/2)^{i}|i\in\Zed^{\pm}\}$, записанная друг под другом так, что цифры одинаковых разрядов находятся в одном столбце, образует слово $u\in\Univ$. Такие $2$-мерные слова $u$ назовём вещественными $2$-мерными $Z$-словами. Вещественное $2$-мерное $Z$-слово $u\in\Univ$ характеризуется тем, что 1) слово $u[0:,-1]$ состоит из символов алфавита $\{0,1,2\}$ и 2) (условие вещественности) любая строка $u[i,:]$, $i\in\Zed^{\pm}$, начиная с некоторого места влево состоит из одних нулей. Если условие 2) не выполняется, то мы будем говорить не о вещественном $Z$-слове, а просто о $Z$-слове.

Продукционные пары, крайние правые столбцы которых не содержат символов $3$, $4$ и $5$, назовём $012$-продукционными парами. Обозначим множество $1$-мерных $012$-продукционных пар через $\I_{012}$.
\kurglemma
\begin{lemma}[\kurglemmaa]
Множество вертикально нетупиковых слов в $\aB\circ\I_{012}$ имеет вид:
{\scriptsize
\[
\begin{aligned}
\begin{bmatrix}
3&2\\
\hline
2&0\\
2&1
\end{bmatrix}
\begin{bmatrix}
0&0\\
2&0\\
\hline
0&0
\end{bmatrix}
\begin{bmatrix}
0&1\\
2&1\\
\hline
3&2
\end{bmatrix}
\begin{bmatrix}
0&0\\
0&1\\
2&0\\
2&1\\
\hline
0&1
\end{bmatrix}
\end{aligned}
\]
}
\end{lemma}
\kurgcorollary
\begin{corollary}[\kurgcorollarya]
Если $\{\{\xi(3/2)^{i}\}|i\in\Zed^{+}\}\subseteq [0,1/2)$, то для $\zeta=\xi\cdot 10^{-1}$ верно $\{\{\zeta(3/2)^{i}\}|i\in\Zed^{+}\}\subseteq [0,1/6)\cup[1/3,2/3)$.
\end{corollary}
\kurgcorollary
\begin{corollary}[\kurgcorollarya]
Если не существует $\xi\in\Real\setminus \{0\}$ такого, что $\{\xi(3/2)^{i}|i\in\Zed^{+}\}\subseteq [0,1/6)\cup[1/3,2/3)$, то множество $Z$-чисел пусто.
\end{corollary}


\begin{figure}
\centering
\includegraphics[width=0.3\textwidth]{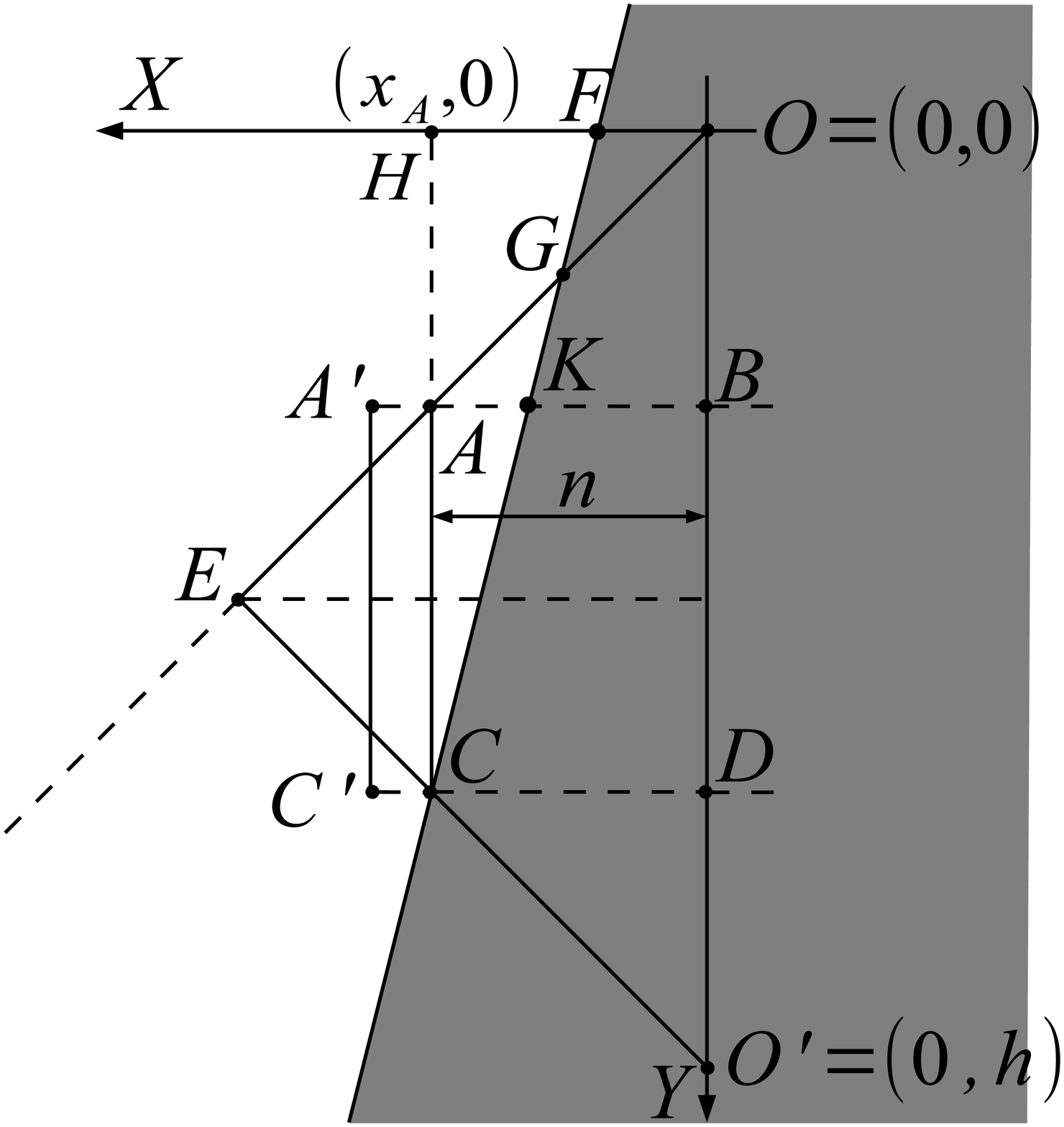}
\caption{Геометрическая иллюстрация необходимого условия $|AC|\le\theta(n)\le\eta(n)$ существования $Z$-чисел.}
\label{figure1.eps}
\end{figure}

Пусть $u:\Zed^2\rightarrow \Zed_6$  вещественное $Z$-слово  и $\xi=u[0,:]\in\Real$.
Слово $u$ изображено на рисунке~\ref{figure1.eps}.
Дискретная решётка системы координат $XY$ представлена в непрерывной форме. Отрезок $OF$ демонстрирует целую часть числа $\xi$. Луч по оси $X$ влево от $F$ соответствует цифрам $0$. 
Координата $x$ крайней левой отличной от $0$ цифры числа 
$
\left(\frac{3}{2}\right)^y\xi
$
равна 
$
\lfloor\log_6((\frac{3}{2})^{y}\xi)\rfloor \approx y\log_6\frac{3}{2}+\log_6\xi
$.
Прямая $x=y\log_6\frac{3}{2}+\log_6\xi$ показана как прямая $FC$. Не закрашенная часть рисунка выделяет фрагмент слова $u$, состоящий из $0$.
Выберем на оси $X$ левее точки $F$ произвольную точку $H=(x_A,0)$. Вертикальная прямая $HA$ пересекает прямую $FC$ в точке $C=(x_A,\frac{1}{\log_6\frac{3}{2}}x_A-\frac{\log_6\xi}{\log_6\frac{3}{2}})$. Прямая $OA$ задана уравнением $y=x$.  Прямая $O'C$ имеет вид $y=-x+c$, $c=\left(\frac{1}{\log_6\frac{3}{2}}+1\right)x_A-\frac{\log_6\xi}{\log_6\frac{3}{2}}$. 
Отрезок $OO'$ обозначает столбец $u[0:h,0]$.
$AC$ является крайним левым столбцом на его горизонтальном уровне, не состоящим из одних нулей.
Обозначим через $A'C'$ столбец, стоящий непосредственно рядом слева от столбца $AC$, то есть столбец $u[x_A:h-x_A+1,x_A+1]$.
Расстояние $|A'B|$ от $A'C'$ до $OO'$ равно $x_A+1$. Высота $|A'C'|$ столбца $A'C'$ равна 
$
|HC|-x_A+O(1)=
\left(\frac{1}{\log_6\frac{3}{2}}-1\right)x_A+O(1)
$.

Столбец $A'C'$ есть результат отображения $L$ столбца $AC$, расширенного сверху и снизу по одному символу. Но столбец из нулей также получается отображением $L$ столбца из нулей, значит здесь $L$ не является инъективным и ситуация на рисунке~\ref{figure1.eps} по следствию~\ref{cor:theta} возможна, если только $|A'C'|<\theta_n$, где $n=|A'B|$.
\kurgcorollary
\begin{corollary}[\kurgcorollarya]
Если
$
\theta(n) < \left(\frac{1}{\log_6\frac{3}{2}}-1\right)n+O(1) \approx 3.42n
$,
то множество $Z$-чисел пусто.
\end{corollary}

\kurgtheorem\begin{theorem}[\kurgtheorema]\label{theorem:mainsufficient}
Множество $Z$-чисел пусто, если, и только если, в вертикально нетупиковых словах в $\aB^{+\omega\circ}\circ\I_{023}$ нет строк из языка $0^{+\omega}\circ\Zed_6^{*}$, где $\Zed_6^{*}$~-- конечные слова в алфавите $\Zed_{6}$.
\end{theorem}
\begin{proof}
Достаточность очевидна, если посмотреть на луч $OG$ на рисунке~\ref{figure1.eps}. Необходимость выводится из следствия~\ref{corollary:u0235w_in_Univ}.
\end{proof}
\kurgcorollary\begin{corollary}[\kurgcorollarya]\label{cor:anothersufficient}
Множество $Z$-чисел пусто, если биекция $\Bi$ (следствие~\ref{cor:1to1}) обладает свойством: если слово $w\in\Zed_6^{+\omega}$ начиная с некоторого момента является периодическим, то в слове $\Bi(w)$ все цифры $\{0,2,3,5\}$ распределены равномерно.
\end{corollary}
\begin{proof}
Пусть $\Bi$ обладает указанным свойством. Луч $OG$ на рисунке~\ref{figure1.eps} с некоторого момента состоит из одних $0$, то есть периодический. Значит в луче $OO'$ есть цифра $5$.
\end{proof}


\end{document}